\documentclass{iopjournal}

\usepackage{amsmath}     % Often needed for math symbols/environments

\usepackage{amssymb,graphicx,tikz,xcolor,colortbl,booktabs,multirow,adjustbox, amsthm}

\newtheorem{proposition}{Proposition}
\newtheorem{definition}{Definition}

\usepackage[justification=justified,singlelinecheck=false]{subcaption}
\usepackage[utf8]{inputenc}
\usepackage[T1]{fontenc}
\usepackage{booktabs,makecell,graphicx}    % Better caption handling
\usepackage[abs]{overpic}
\usepackage[T1]{fontenc}
\usepackage[utf8]{inputenc}

\usepackage{graphicx,dcolumn,bm,hyperref}% add hypertext capabilities
\usepackage[mathlines]{lineno}% Enable numbering of text and display math

\usepackage{ragged2e} % Required for better justification
\usepackage{caption}

\usepackage[english]{babel}
\usepackage{microtype}

\usepackage{float}
\usepackage{algorithm, algpseudocode}
\usepackage{algorithmicx}
\usepackage{dsfont} % for 1 math symbol equation 9

\DeclareCaptionJustification{justified}{\justifying} % Force 'justified' to use ragged2e
\def\bea{\begin{eqnarray}}
\def\eea{\end{eqnarray}}

\makeatletter
\g@addto@macro\appendix{%
  \setcounter{section}{0}%
  \setcounter{equation}{0}%
  \@addtoreset{equation}{section}% Automatically resets equation count per appendix section
}

\usepackage{array}
\usepackage{tabularx}
\usepackage{multirow}
\usepackage{makecell}

\usetikzlibrary{arrows.meta, positioning, fit, backgrounds}

\newcolumntype{Y}{>{\centering\arraybackslash}X}
\newcommand{\Cx}{\ensuremath{C\mathbf{x}=\mathbf{b}}}
\newcommand{\syn}{\textsc{SUTURE}}
\newcommand{\purify}{\textsc{Purify}}
\newcommand{\ER}{\ensuremath{\mathbb{E}[1/R]}}
\newcommand{\lam}{\ensuremath{\lambda}}

\begin{document}

\articletype{Paper} %	 e.g. Paper, Letter, Topical Review...

\title{SUTURE: Syndrome-Guided Repair for Segmented Feasibility-Preserving VQAs on Noisy Hardware}

\author{Sokea Sang$^1$\orcid{0009-0006-0345-5870}, Leanghok Hour$^2$\orcid{0009-0001-5670-7481}, Sanghyeon Lee$^1$\orcid{0000-0002-7258-0521} and Youngsun Han$^{1,*}$\orcid{0000-0001-7712-2514}}

\affil{$^1$Department of AI Convergence, Pukyong National University, Busan 48513, Republic of Korea}

\affil{$^2$AI Research Institute, Pukyong National University, Busan 48513, Republic of Korea}

% \affil{$^2$Department, Institution, City, Country}

% \affil{$^*$Author to whom any correspondence should be addressed.}

\email{youngsun@pknu.ac.kr}


\justifying

\begin{abstract}
Constrained binary optimization is a representative class of NP-hard problems in scheduling, resource allocation, and finance. Segmented feasibility-preserving variational quantum algorithms (VQAs) are a promising approach that restricts ideal circuit evolution to feasible assignments and executes short measure-and-reseed segments. The existing boundary runtime enforces feasibility through \emph{purification}, which discards measurements that violate the constraints. As problem size and noise increase, feasible measurements become rare; if no shot survives, the execution chain terminates. In our 72-qubit graph-coloring experiment on a real IBM Heron device, 11 of 12 purification chains terminate before completing all segments.
We propose \emph{SUTURE}: syndrome-guided repair, a runtime that repairs infeasible measurements instead of discarding them and is deployable on current quantum devices.
SUTURE exploits a parity-check-like structure induced by the problem constraints: violated constraints form a \emph{syndrome} that detects and localizes corruption and, under suitable structural conditions, often identifies a correction. SUTURE combines three stages: (1) a compile-time profiler that uses the compiled constraints and feasible initialization samples to predict single-flip recovery, with a maximum error of 0.011 across four constraint families; (2) a bounded runtime decoder that replaces purification inside the segmented execution loop; and (3) a regime analysis that identifies when repair is preferable to purification. In simulation up to 120 qubits, \syn{} continues the execution chain far beyond the noise level at which purification collapses. 
In the same 72-qubit experiment on IBM Heron hardware, SUTURE completes all segments in all 12 runs. In a separate 36-qubit timing experiment, decoding adds 1.6\% to execution-path latency and less than 1\% to total pipeline latency.
\end{abstract}

%% =====================================================================
\section{Introduction}
\label{sec:intro}

% Constrained binary optimization appears in scheduling, resource
% allocation, and finance. A solver must find a binary assignment that
% satisfies hard constraints while optimizing an objective. Variational
% quantum algorithms (VQAs) commonly encode violations as objective
% penalties, but this approach allows the circuit to explore the entire
% binary space and can place substantial measured probability mass on
% invalid assignments. Feasibility-preserving VQAs instead construct circuit
% transitions that remain within the feasible subset. Rasengan~\cite{rasengan}
% uses this approach in a segmented execution model: it runs a short
% transition circuit, measures the register, and uses the surviving
% measurements to seed the next segment. In the absence of noise, the
% construction preserves feasibility. On hardware, the boundary runtime
% enforces it through \emph{purification}, which discards every measured
% bitstrings that violates a constraint.
In combinatorial optimization, constrained binary optimization
problems seek the optimal value of decision variables that can
only take binary values, subject to a system of constraints~\cite{integer_programming}. These
problems arise widely in resource allocation~\cite{register_allocation_gcp}, route
planning~\cite{apply_qaoa_routing}, engineering design~\cite{qauntum_op_satellite_mission}, and financial portfolio selection~\cite{benchmarking_portfolio},
and they are typically NP-hard~\cite{analyze_NP_hard}: the cost of finding the exact optimum on classical computers grows exponentially with problem size. 
Quantum computing can represent and transform many binary configurations through superposition and entanglement, and quantum algorithms have demonstrated theoretical advantages for certain computational tasks~\cite{quantum_measurements_abelian, shor_algorithm}. In particular, as a class of hybrid quantum-classical algorithms, variational quantum algorithms (VQAs)~\cite{variational_alg_review, qoaoa_nisq, frozenLGP, sang2026landscape} use parameterized quantum circuits together with a classical optimizer to exploit this capability for combinatorial optimization. For constrained problems, however, conventional penalty-based VQAs~\cite{penalty_qubo_formulation, penalty_partition_qubo, driver_hamiltonians_contrain_op, hea_vqa} explore the full set of $2^n$ binary assignments, in which feasible solutions may form only a small fraction. Constraint-preserving approaches instead restrict the ideal quantum dynamics to the feasible subspace, avoiding measurement probability being intentionally assigned to invalid solutions~\cite{chocoq}.

Rasengan~\cite{rasengan} realizes this principle through transition Hamiltonians derived from the null space of the compiled constraint matrix $C$ (In linear algebra, a single feasible solution can span the
full solution space via homogeneous basis vectors~\cite{rasengan}). Starting from one feasible seed, each transition moves amplitude between feasible basis states while preserving $C\mathbf{x}=\mathbf{b}$ under noise-free execution. To make the resulting transition sequence deployable on current hardware, Rasengan divides it into short segments and measures and reinitializes the register between them. This segmentation creates an important runtime boundary: hardware noise can produce infeasible measurements, and Rasengan's purification policy discards them before initializing the next segment. When no feasible measurement survives, the execution chain terminates.

Table~\ref{tab:runtime_comparison} highlights the resulting runtime gap. Purification acts at every segment boundary but can only discard infeasible measurements, causing the execution chain to terminate when no feasible shot survives. Conventional repair methods~\cite{sqd, sqd-codespace, regrid-qaoa, ising-bench, chainbreak-qst,qac-me} instead operate after the end of the execution loop and therefore cannot influence subsequent segments. A suitable boundary runtime must repair infeasible measurements under a bounded classical budget, reinject the repaired states, and preserve the current execution trajectory. This distinction motivates an in-loop repair mechanism rather than another terminal post-processing method.

\begin{table}[t]
\centering
\caption{Feasibility handling in segmented VQAs.}
\label{tab:runtime_comparison}

\footnotesize
\setlength{\tabcolsep}{2.3pt}
\renewcommand{\arraystretch}{1.12}
\renewcommand{\tabularxcolumn}[1]{m{#1}}

\begin{tabularx}{\columnwidth}{
    |>{\centering\arraybackslash}m{0.19\columnwidth}
    |Y|Y|Y|Y|}
\hline

\multirow{2}{*}{\makecell{\textbf{Design}\\\textbf{property}}}
& \multicolumn{2}{c|}{\textbf{Existing handling}}
& \multicolumn{2}{c|}{\textbf{\syn{}}} \\
\cline{2-5}

& \textbf{Purify~\cite{rasengan}}
& \makecell{\textbf{Terminal}\\\textbf{repair}}
& \textbf{\syn{}$_C$}
& \textbf{\syn{}$_G$} \\
\hline

\textbf{Execution point}
& \makecell{Boundary\\(segment)}
& \makecell{Terminal\\(post-pro.)}
& \makecell{Boundary\\(segment)}
& \makecell{Boundary\\(segment)} \\
\hline

\textbf{Invalid result}
& Discard
& \makecell{Repair\\(end exec.)}
& \makecell{Repair\\(segment)}
& \makecell{Repair\\(segment)} \\
\hline

\textbf{Reinjection}
& Yes
& No
& Yes
& Yes \\
\hline

\textbf{Repair input}
& No
& Varies
& $C,\mathbf{b}$
& $C,\mathbf{b},G$ \\
\hline

\textbf{Objective-blind}
& N/A
& Varies
& Yes
& No \\
\hline

\textbf{Chain liveness}
& \makecell{Terminate\\under noise}
& No chain
& \multicolumn{2}{c|}{
    \makecell{Robust to noise,\\continues liveness}
  } \\
\hline

\end{tabularx}

\begin{minipage}{\columnwidth}
\footnotesize\noindent
\textit{
Purification discards infeasible measurements and may terminate
the chain when no feasible shot survives. \syn{} repairs and reinjects
measurements during execution; \syn{}$_C$ uses only the constraint
system, whereas \syn{}$_G$ additionally uses graph adjacency.
Unlike \syn{}, terminal repair is primarily suited to traditional
monolithic QAOA or VQA workflows, where it improves measurements
only after the complete circuit execution. \syn{} instead repairs
and reinjects infeasible measurements at intermediate boundaries,
allowing a segmented chain to remain live.
}
\end{minipage}

\end{table}

\begin{figure}[tbp]
  \centering
  \includegraphics[width=0.75\linewidth]{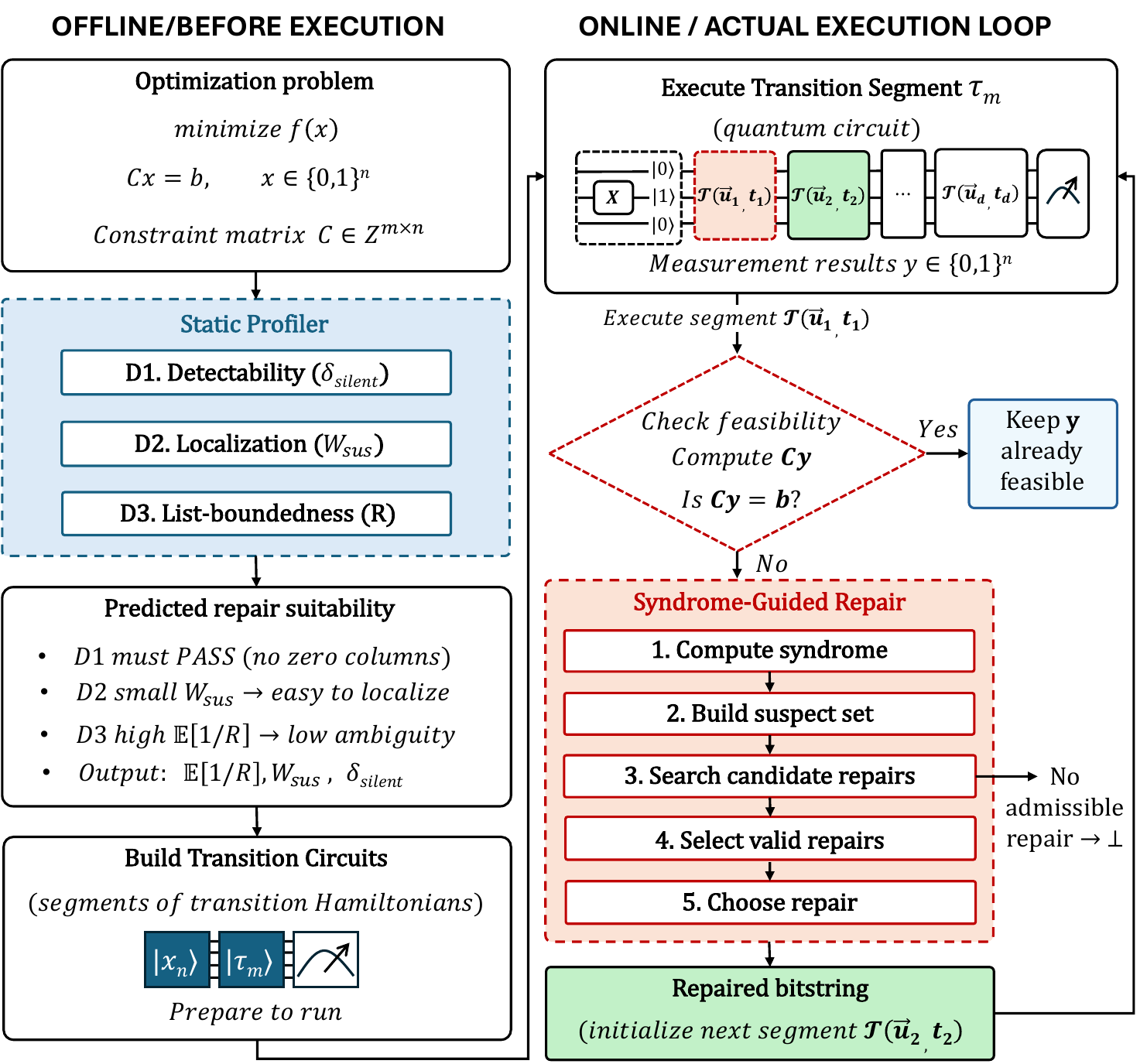}
  \caption{Overview of the \syn{} pipeline. Stage 1 profiles repairability from the compiled constraints before execution; Stage 2 localizes and repairs infeasible measurements inside the segmented execution loop. 
  %Stage 3 regime characterization is offline and therefore not shown.
  }
  \label{fig:pipeline}
\end{figure}

Purification makes feasibility a liveness concern. Suppose a segment
produces $S$ shots and each shot is feasible with probability $\phi_t$.
The segment retains only $S\phi_t$ shots in expectation, and the chain
cannot continue if it retains none. In the structured encodings we
study, feasible assignments occupy an exponentially small fraction of
the binary space, so $\phi_t$ falls as problem size and noise increase.
Failure then compounds across segment boundaries. In our 72-qubit IBM
Heron experiment with $S=64$, 11 of 12
purification chains terminate before completing 96 segments; reducing
the budget to 32 shots moves the median termination point from segment
76 to segment 4. The device continues to produce measurements, but the
boundary runtime discards too many of them to sustain execution.

We observe that those discarded measurements still contain problem-specific
information. Each infeasible bitstring reveals exactly which constraints
it violates. We call this set of violations its \emph{syndrome}. The
syndrome identifies the variables involved in the failed constraints,
which often reduces a global correction problem to a small local
search. This signal requires no additional qubits or measurements: it
comes from the constraint system already available to the compiler.
Moreover, each segmented transition touches only a small active set, so
violations often remain localized to one or two encoding blocks, a
property we measure directly on hardware. This observation motivates
our central question: \emph{can a boundary runtime use the problem's
own constraints to recover useful feasible states from the
measurements that purification discards, and when should it do so?}

We answer this question with \emph{\syn{}} (summarized in Figure~\ref{fig:pipeline} and detailed in Section~\ref{sec:method}), a syndrome-guided repair
runtime for segmented feasibility-preserving VQAs. Before execution, a
static profiler checks whether the constraint encoding makes corruption
detectable, localizable, and sufficiently unambiguous to repair. During
execution, an objective-blind decoder uses each syndrome to search only
the implicated variables, returns a feasible correction when one is
found within a fixed budget, and reinjects the corrected bitstrings into the
chain. Finally, an offline regime analysis separates two deployment
decisions: the quality crossover $\lambda_q$, above which repair wins
on quality while both policies remain live, and the purification survival cliff
$\lambda_s$, beyond which discard is no longer a reliable boundary
policy. This distinction matters because repair is not uniformly
better: below a measured crossover, purification remains the
conservative choice.

% % \textbf{Contributions.}
% This paper makes three contributions. First, we introduce a static
% repairability profiler and validates its single-flip accuracy law across
% four constraint families. Second, we develop \syn{}$_C$, an
% objective-blind repair mechanism that operates inside a segmented VQA
% rather than as terminal post-processing, together with measured-only
% endpoints and a contamination screen for attribution; we also develop
% \syn{}$_G$, an explicitly graph-aware specialization for soft-edge
% encodings. Third, we provide a two-boundary regime analysis on
% simulators and real hardware that identifies both where repair improves
% execution and where purification remains preferable.
Overall, this paper makes the following contributions:
\begin{enumerate}
    \item We introduce a static repairability profiler and validate its single-flip accuracy law across four constraint families.

    \item We develop \syn{}$_C$, an objective-blind repair mechanism that, to our knowledge, is the first deployable repair mechanism to operate inside a segmented VQA execution loop rather than as terminal post-processing.
    We pair it with measured-only endpoints and a contamination screen for rigorous attribution. We also develop \syn{}$_G$, an explicitly graph-aware specialization for soft-edge encodings.

    \item We provide a deployment regime analysis using simulators and real
    quantum hardware, identifying when repair improves execution, when purification remains preferable, and where repair reaches its own liveness limit.
\end{enumerate}

% We evaluate 109 instance packs spanning four constraint families such as graph coloring problem (GCP)~\cite{graph_coloring_gcp}, facility location problem (FLP)~\cite{facility_fcp}, k-partition problem (KPP)~\cite{k_partition_kpp} and set cover problem (SCP)~\cite{set_cover_scp} with problem sizes ranging from
% 15 to 120 qubits. The profiler predicts measured single-flip recovery
% across 70 instances with a maximum error of 0.011. In
% contamination-screened hard-edge graph-coloring tests, objective-blind
% \syn{}$_C$ separates from purification and its continuation baselines
% at 78 and 90 qubits in the post-cliff regime. The graph-aware
% \syn{}$_G$ specialization extends the soft-edge scale study to 120
% qubits. On IBM Heron hardware (ibm\_marrakesh)~\cite{IBMQuantum2026}, \syn{} improves quality while both
% runtimes are live at 36 and 48 qubits and preserves liveness at 72
% qubits; in the 36-qubit timing experiment, decoding contributes less
% than 1\% of total pipeline latency. 

% update 
We evaluate \syn{} across 109 benchmark packs from four constraint families graph coloring (GCP)~\cite{graph_coloring_gcp}, facility location (FLP)~\cite{facility_fcp}, $k$-partition (KPP)~\cite{k_partition_kpp}, and set cover (SCP)~\cite{set_cover_scp}, including small cross-family instances down to 8 qubits and scale ladders extending to 120 qubits. Across 70 instances, the static profiler predicts measured single-flip recovery with a maximum error of only 0.011. In contamination-screened hard-edge GCP experiments at 78 and 90 qubits, objective-blind \syn{}$_C$
completes all segments and improves measured proper-coloring success
over purification. The graph-aware \syn{}$_G$ specialization
extends the soft-edge scaling study to 120 qubits. On IBM Heron hardware~\cite{IBMQuantum2026}, \syn{}$_G$ improves solution quality at 36 and 48 qubits; at 72 qubits, all \syn{}$_G$ chains complete while 11 of 12 purification chains terminate. Finally, decoding adds only 1.6\% of execution-path latency in our 36-qubit timing experiment, or less than 1\% of total pipeline latency.
These experiments compare boundary runtimes under matched segmented-VQA workloads, we do not claim quantum advantage over classical solvers such as DSATUR~\cite{dsatur}. 

\section{Background}
\label{sec:background}
\label{sec:bg-cop}
\label{sec:bg-vqa}

\textbf{Problem model.}
At a high level, constrained binary optimization chooses a binary
assignment, rejects assignments that violate the constraints, and uses
an objective to rank the remaining feasible assignments. We represent
the hard constraints as linear equalities; inequalities can be
converted to this form with binary slack variables. The resulting
problem is
\begin{equation}
\min_{\mathbf{x}}\; f(\mathbf{x}),
\qquad \text{s.t.}\; C\mathbf{x}=\mathbf{b},\;
\mathbf{x}\in\{0,1\}^{n}.
\label{eq:cop}
\end{equation}
Here $\mathbf{x}$ contains $n$ binary variables, $C$ is an
$m\times n$ integer constraint matrix, $\mathbf{b}$ is the required
constraint value, and $f$ is the objective. The constraints define the
feasible set $F=\{\mathbf{x}:C\mathbf{x}=\mathbf{b}\}$; the objective
ranks only the assignments in $F$. Our running example is graph
coloring (GCP). A one-hot block of $K$ bits encodes the color of each
vertex, and an assignment row requires exactly one selected color per
vertex. In the soft-edge encoding, the objective counts conflicts
between adjacent vertices; in the hard-edge encoding, those edge
requirements instead enter $C$ through slack variables. Facility
location (FLP), set cover (SCP), and
$k$-partition (KPP) use the same linear form but have different row
density and slack structure. Section~\ref{sec:method-d123} shows how
these structural differences determine repairability.

\textbf{Feasibility-preserving VQAs.}
Penalty-based VQAs, including penalty QAOA~\cite{penalty_qubo_formulation, driver_hamiltonians_contrain_op, hea_vqa, penalty_partition_qubo}, allow the
circuit to explore all $2^n$ binary assignments and penalize constraint
violations through the objective. Noise can therefore move measured
probability far outside $F$. Constraint-preserving approaches instead
engineer the circuit dynamics to remain within $F$. Examples include
commuting-Hamiltonian constructions and constrained mixers such as XY
mixers~\cite{xy-mixers} and Choco-Q~\cite{chocoq}.
Rasengan~\cite{rasengan} builds transition Hamiltonians from directions
that do not change the constraint value. Specifically, a transition
direction $u$ lies in the null space of $C$, so $Cu=0$, and the
resulting ideal transition preserves $C\mathbf{x}=\mathbf{b}$. The
corresponding circuit touches only the qubits where $u$ is nonzero
(Figure~\ref{fig:transition}).

\textbf{Segmented execution and purification.}
Rasengan applies these transitions as a chain of short segments. After
each segment, the register is measured, and the resulting classical
bitstrings seed the next segment. This measurement boundary is what
makes a per-shot classical policy possible without altering the quantum
circuit. The incumbent policy is purification: it retains bitstrings in
$F$ and discards bitstrings that violate \Cx{}. Without noise, every
measurement is feasible and the filter does nothing. With noise, a
segment may produce no feasible measurement; without a replacement
state, the chain terminates. Because this risk occurs at every boundary,
segmentation turns an output-validity problem into an execution-liveness
problem. Section~\ref{sec:method} develops a boundary runtime that uses
the failed constraints as a recovery signal instead of discarding the
entire measurement.

\begin{figure}[t]
  \centering
  \includegraphics[width=0.95\linewidth]{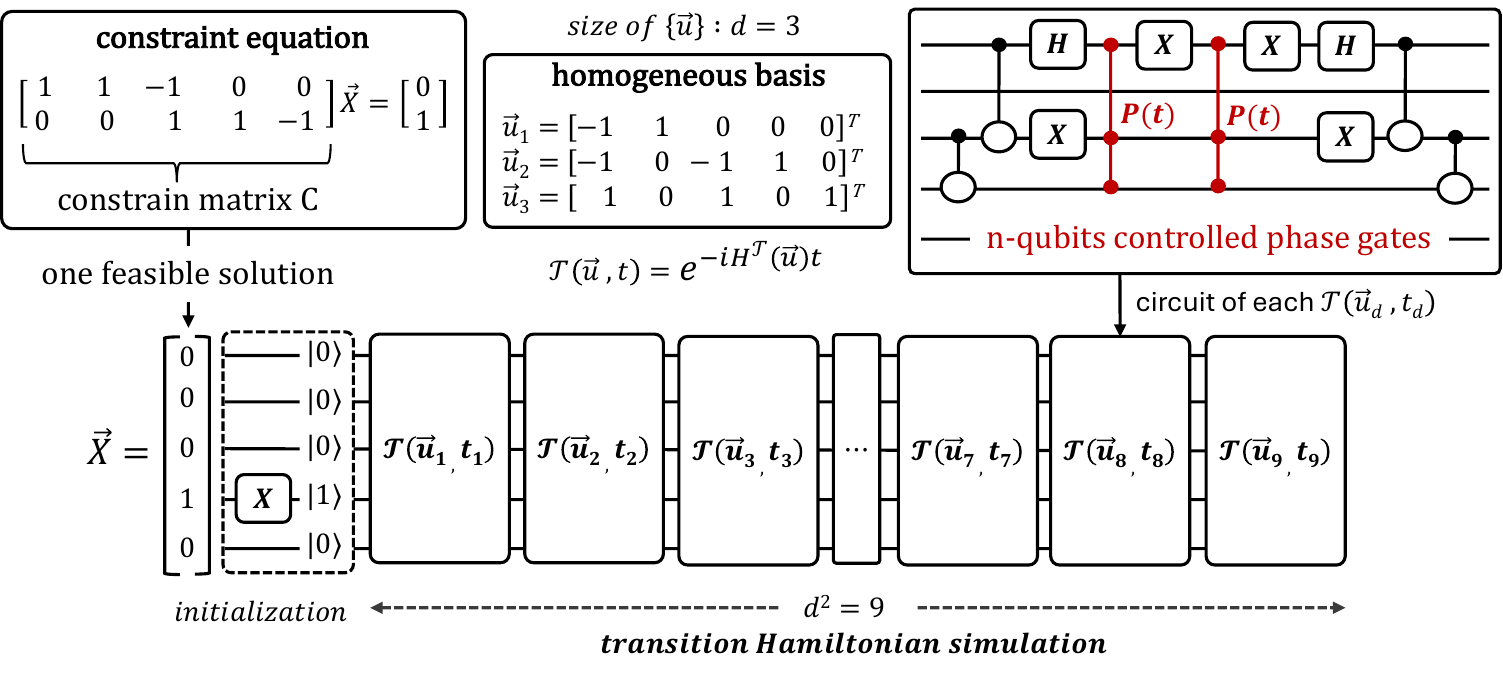}
  \caption{Rasengan segmented execution, adapted from~\cite{rasengan}. A vector \(u\) satisfying
  \(Cu=0\) defines a transition circuit whose active set contains the
  qubits at the nonzero positions of \(u\). Measurement between segments
  enables purification or \syn{} to act at each boundary.}
  \label{fig:transition}
\end{figure}

\section{\syn{}: Syndrome-Guided Repair}
\label{sec:method}

Figure~\ref{fig:pipeline} summarizes \syn{}. At compile time, a static
profiler examines the constraint system and available feasible
initialization samples to estimate whether a corruption can be detected,
localized, and repaired unambiguously. During execution, a per-shot
decoder identifies the constraints violated by each measured bitstring and
searches only the implicated variables for a feasible repair. A third,
offline stage determines when repair should replace purification as noise
increases. We first define the segmented execution model and its boundary
runtime, then present the profiler, decoder, and deployment regimes.

\subsection{Setting: segmented feasibility-preserving chains}
\label{sec:method-setting}

We inherit Rasengan's circuit layer unchanged~\cite{rasengan}. Let $u$ be
a null-space vector of $C$, so $Cu=0$, and let
$A(u)=\{j:u_j\neq0\}$ be its \emph{active set}. Rasengan constructs a
transition Hamiltonian $H_u$ from $u$, exponentiates it, and decomposes
the result into a short circuit of one- and two-qubit gates acting only on
$A(u)$. In one-hot GCP, for example, a transition swaps two color entries
of one vertex and, in the hard-edge model, updates the associated slack
bits. A chain of these transitions, with trained angles, expands the
solution distribution. The register is measured between segments, and
the measured bitstrings seed the next segment.

Two inherited properties enable \syn{}. First, each transition is local:
it touches only $A(u)$, so local noise tends to produce violations in a
small number of encoding blocks. We test this observable locality under
device-calibrated noise and on hardware
(Section~\ref{sec:eval-device}). Second, measurement classicalizes the
register at every boundary. The inter-segment state is therefore a
distribution over bitstrings on which a per-shot classical decoder can
act without modifying the quantum circuit.

\subsection{Problem formulation}
\label{sec:method-formal}

We next define the runtime's input, output, and observable repair signal.
The optimization problem itself is unchanged from
Section~\ref{sec:bg-cop}; \syn{} changes only the policy applied at a
segment boundary.

We use the instance $(f,C,\mathbf{b})$ and feasible set $F$ defined in
Equation~\eqref{eq:cop}. As in Rasengan~\cite{rasengan}, the compiler
converts inequalities to equalities with binary slack variables, which
are included in $\mathbf{x}$.

\begin{definition}[Segmented execution]
\label{def:chain}
A segmented chain of length $T$ with shot budget $S$ maintains a
\emph{register} $X_t$, a multiset of $S$ feasible bitstrings, initialized
to $S$ copies of a feasible seed $\mathbf{x}_0$. Segment $t$ applies a
short unitary $U_t$ to the basis states represented by $X_t$ and then
measures $S$ shots, producing the multiset
$Y_t=\{\mathbf{y}^{(1)},\dots,\mathbf{y}^{(S)}\}$ of binary bitstrings.
Let $\mathcal{H}_F=\operatorname{span}\{|\mathbf{x}\rangle:
\mathbf{x}\in F\}$ be the feasible subspace. By construction,
$U_t\mathcal{H}_F\subseteq\mathcal{H}_F$, so every noise-free
measurement belongs to $F$. Under a noise channel
$\mathcal{N}_\lambda$ with scale $\lambda$, measurements may leave
$F$.
\end{definition}

\begin{definition}[Boundary runtime]
\label{def:runtime}
A \emph{boundary policy} maps the current measured multiset and,
optionally, the runtime history $H_t$ (earlier registers, the feasible
seed) to a multiset of feasible bitstrings. When the policy returns between
one and $S-1$ bitstrings, the runtime resamples them with replacement to
form the $S$-bitstring register $X_{t+1}$. A policy is \emph{history-free}
when it uses only $Y_t$, in which case we write $X_{t+1}=\Pi(Y_t)$.
The chain \emph{dies at $t$} under $\Pi$ if the policy returns no
bitstrings. The incumbent policy is \emph{purification},
$\Pi_{\mathrm{P}}(Y)=Y\cap F$, which retains every feasible occurrence
in $Y$ and discards the rest. \syn{} is also history-free. The
purify+restart and purify+last baselines instead consult $H_t$ when
purification returns no survivors; Section~\ref{sec:eval-setup} defines
these fallbacks.
\end{definition}

Purification's failure mode depends on the \emph{feasible fraction}
$\phi_t=\Pr[\mathbf{y}\in F]$ under segment $t$'s measurement
distribution. Assuming $S$ independent shots with this common feasible
probability,
\begin{equation}
\Pr[\text{chain dies at }t]=(1-\phi_t)^S,
\label{eq:cliffprob}
\end{equation}
which rises sharply once $S\phi_t=O(1)$ (the per-segment 50\% death
point is $\phi_t=1-2^{-1/S}\approx\ln 2/S$); we use $\phi_t\approx 1/S$
as the engineering shorthand for this local threshold. Stage~3 later
defines the corresponding chain-level survival boundary.

An infeasible bitstring reveals more than the fact that purification would
discard it: it identifies the constraints that failed. This observable
signal drives both the profiler and the decoder.

\begin{definition}[Syndrome]
\label{def:syndrome}
Let $C_r$ denote row $r$ of $C$, and let $\operatorname{supp}(C_r)$ be
the indices of its nonzero entries. The \emph{syndrome} of a measured
bitstring $\mathbf{y}$ is the set of violated constraint rows. Its
\emph{suspect set} is the union of the variables touched by those rows:
\begin{equation}
\begin{aligned}
\sigma(\mathbf{y})
  &=\{r\in\{1,\dots,m\}:(C\mathbf{y})_r\neq b_r\},\\
L(\sigma(\mathbf{y}))
  &=\bigcup_{r\in\sigma(\mathbf{y})}\operatorname{supp}(C_r).
\end{aligned}
\label{eq:syndrome}
\end{equation}
Thus $\mathbf{y}\in F$ if and only if
$\sigma(\mathbf{y})=\varnothing$.
\end{definition}

For an encoding partitioned into decision blocks
$B_1,\dots,B_q$, a block is \emph{suspect} when it intersects
$L(\sigma(\mathbf{y}))$. We write
$c(\mathbf{y})=|\{i:B_i\cap L(\sigma(\mathbf{y}))\neq\varnothing\}|$
for the number of suspect blocks. In soft-edge one-hot GCP, where $C$
contains the assignment rows, this is exactly the number of blocks whose
measured one-hot weight is not one.

The runtime problem can now be stated without committing to a particular
decoder:

\begin{quote}
\emph{Given $Y_t$, $C$, and $\mathbf{b}$, produce a feasible next
register without inspecting the objective $f$. The policy must terminate
under a fixed per-shot search budget and should keep the chain live when
purification has no survivors, provided that at least one measurement
admits a repair within that budget.}
\end{quote}

Objective-blindness is essential for attribution: a policy that reads
$f$ is itself a classical optimizer, so its output quality cannot be
credited solely to the quantum computation
(Section~\ref{sec:eval-setup}). The design also relies on two premises.
\emph{A1 (classicalization)} holds by construction: $Y_t$ contains
classical bitstrings, not an unmeasured quantum state. \emph{A2
(observable locality)} is empirical: local transitions and predominantly
local device noise should keep $c(\mathbf{y})$ small. Unlike the latent
error $\mathbf{e}=\mathbf{y}\oplus\mathbf{x}$ relative to an unknown
clean bitstring $\mathbf{x}$, both $\sigma(\mathbf{y})$ and
$c(\mathbf{y})$ are available to the runtime. The block-structured
decoders rely directly on A2; the generic decoder relies on the analogous
condition that the suspect set admits a repair within a small flip
budget. A corruption can still be syndrome-silent when it maps one valid
block codeword to another; Stage~1 measures this loss of detectability, and
~\ref{app:eval-ablate} shows why one-hot encoding is preferable
for repair. Section~\ref{sec:eval-device} tests A2 directly on hardware.

\subsection{Stage 1: static repairability (D1--D3)}
\label{sec:method-d123}

Not every problem can be repaired accurately, and a deployer should know
this \emph{before} spending quantum time. Stage 1 answers that question
statically from $C$, $\mathbf{b}$, and feasible samples already
available during initialization. Let $C_j$ denote column $j$ of $C$,
and let $\mathbf{e}_j$ be the binary vector with a one only at position
$j$. Let $Q_F$ denote the initialization procedure's distribution over
feasible bitstrings. To profile a single-bit corruption, draw
$\mathbf{X}\sim Q_F$ and a position $J$ uniformly from
$\{1,\dots,n\}$, then set $\mathbf{Y}=\mathbf{X}\oplus\mathbf{e}_J$.

The profiler reports three quantities. D1 measures whether a single
flip is visible to the syndrome. D2 measures how many variables become
suspects after such a flip. D3 counts how many suspect flips would
restore feasibility and therefore measures repair ambiguity:
\begin{equation}
\begin{gathered}
\text{D1 (detection):}\quad
\delta_{\mathrm{silent}} = \tfrac{1}{n}\bigl|\{j : C_j=\mathbf{0}\}\bigr|,
\\
\text{D2 (localization):}\quad
W_{\mathrm{sus}} = \operatorname*{median}_{j}\,
  \bigl|L(\operatorname{supp}(C_j))\bigr|,\\
\text{D3 (ambiguity):}\quad
R(\mathbf{y}) = \bigl|\{ j\in L(\sigma(\mathbf{y})) :
\mathbf{y}\oplus\mathbf{e}_j\in F \}\bigr|,
\end{gathered}
\label{eq:d1}
\end{equation}
where $W_{\mathrm{sus}}$ is the median suspect-set width and
$R(\mathbf{y})$ is the repair degeneracy, or tie-break list size. D1
and D2 depend only on $C$; D3 additionally depends on $Q_F$. The
following propositions give their single-flip guarantees (proofs in
~\ref{app:proofs}).

\begin{proposition}[Detection]
\label{prop:detect}
If no column of $C$ is zero, then for every $\mathbf{x}\in F$ and every
$j$, $\sigma(\mathbf{x}\oplus\mathbf{e}_j)=\mathrm{supp}(C_j)
\neq\varnothing$; hence $\delta_{\mathrm{silent}}=0$.
\end{proposition}

\begin{proposition}[Localization]
\label{prop:localize}
For a single flip at $j$, $j\in
L(\sigma(\mathbf{x}\oplus\mathbf{e}_j))=
\bigcup_{r\in\mathrm{supp}(C_j)}\mathrm{supp}(C_r)$. Consequently, the
suspect-set size is at most the column degree of $j$ times the maximum
row width.
\end{proposition}

\begin{proposition}[Accuracy law]
\label{prop:accuracy}
Fix any distribution over detectable pairs $(\mathbf{x},j)$ with
$\mathbf{x}\in F$ and $C_j\neq\mathbf{0}$. A decoder that draws
uniformly from the feasibility-restoring single-flip candidates recovers
the planted bitstring with probability exactly $\mathbb{E}[1/R(\mathbf{Y})]$.
\end{proposition}

The profiler estimates $\mathbb{E}[1/R(\mathbf{Y})]$ before quantum
execution and reports the silent-flip rate separately. The evaluation
validates this prediction across four constraint families
(Section~\ref{sec:eval-gen}). Computing D1 and D2 requires only a scan
of $C$; estimating D3 from $N$ sampled feasible bitstrings costs
$O(Nnm)$. Across the 27 cross-family packs used for this estimate, the
profiler takes a median of 1.2 seconds and at most 2.7 seconds on one
CPU core. D1 also serves as an encoding screen: one-hot encodings have
no silent single-bit flips, whereas ~\ref{app:eval-ablate}
shows that a more compact domain-wall encoding loses the repair signal.

\subsection{Stage 2: the \syn{} runtime decoder}
\label{sec:method-synd}

Stage 2 turns the syndrome into a repair inside the execution loop. For
each measured bitstring, the decoder computes the syndrome, pins every
variable outside the suspect set, searches the remaining variables under
a fixed budget, and reinjects one feasible repair into the next register.

\begin{definition}[\syn{} decoder]
\label{def:synd}
Let $d_H(\mathbf{y},\mathbf{z})$ denote Hamming distance. For a measured
bitstring $\mathbf{y}$, the ideal \emph{repair set} is
\begin{equation}
\mathcal{R}(\mathbf{y})=
\operatorname*{arg\,min}_{\substack{\mathbf{z}\in F\\
\mathbf{z}_j=\mathbf{y}_j\ \forall j\notin L(\sigma(\mathbf{y}))}}
d_H(\mathbf{y},\mathbf{z}),
\label{eq:repairset}
\end{equation}
the feasible bitstrings closest to $\mathbf{y}$ in Hamming distance among
those that agree with $\mathbf{y}$ outside the suspect set. This
agreement condition is the decoder's \emph{pinning} rule. Let
$\mathcal{S}(\mathbf{y})$ be the candidate set produced by a particular
implementation, and let $\mathrm{Budget}_{\mathcal{S}}(\mathbf{y})$
indicate that its fixed search limit is satisfied. The per-shot decoder
passes feasible bitstrings through unchanged, draws uniformly from a
nonempty candidate set within budget, and otherwise gives up
($\bot$):
% \begin{equation}
% D_{\mathcal{S}}(\mathbf{y})=
% \begin{cases}
% \mathbf{y} & \sigma(\mathbf{y})=\varnothing,\\[2pt]
% \mathbf{z}\sim\mathrm{Unif}(\mathcal{S}(\mathbf{y}))
%   & \mathcal{S}(\mathbf{y})\neq\varnothing,\;
%     \mathrm{Budget}_{\mathcal{S}}(\mathbf{y}),\\[2pt]
% \bot & \text{otherwise.}
% \end{cases}
% \label{eq:decoder}
% \end{equation}
\begin{equation}
D_{\mathcal{S}}(\mathbf{y})=
\begin{cases}
\mathbf{y},
  & \text{if }\sigma(\mathbf{y})=\varnothing,\\[3pt]
\mathbf{z}\sim\mathrm{Unif}(\mathcal{S}(\mathbf{y})),
  & \begin{aligned}[t]
    \text{if }&\sigma(\mathbf{y})\neq\varnothing,\\
              &\mathcal{S}(\mathbf{y})\neq\varnothing,\quad
               \mathrm{Budget}_{\mathcal{S}}(\mathbf{y}),
    \end{aligned}\\[3pt]
\bot,
  & \text{otherwise}.
\end{cases}
\label{eq:decoder}
\end{equation}
\end{definition}
The subscript identifies the information available to candidate
selection. \syn{}$_{C}$ reads only $(\mathbf{y},C,\mathbf{b})$ and is
objective-blind. \syn{}$_{G}$ additionally reads graph adjacency $G$ and
is reported separately. Unsubscripted \syn{} refers only to the shared
mechanism: syndrome localization, pinning, bounded search, and
reinjection.
Generic \syn{}$_{C}$ searches flip sets within
$L(\sigma(\mathbf{y}))$ in increasing Hamming weight. Its budget
$w_{\max}$ is the largest number of flipped positions it will enumerate;
the implementation uses $w_{\max}=6$. 

The fixed-$K{=}3$ GCP decoders
instead use the already-defined suspect-block count $c(\mathbf{y})$ and
give up when it exceeds a fixed block budget $c_{\max}$. 
For each suspect vertex block $v$, they form the candidate-color, or \emph{coset},
list
% \begin{equation}
%   \Lambda(v) =
%   \begin{cases}
%     \{c : y_{vc}=1\} & \text{if } \textstyle\sum_c y_{vc}=2,\\[2pt]
%     \{1,\dots,K\}    & \text{otherwise},
%   \end{cases}
%   \label{eq:cosetlist}
% \end{equation}
\begin{equation}
\Lambda(v)=
\begin{cases}
\{c:y_{vc}=1\},
  & \text{if }\textstyle\sum_c y_{vc}=2,\\[3pt]
\{1,\dots,K\},
  & \text{otherwise}.
\end{cases}
\label{eq:cosetlist}
\end{equation}
% which contains the colors left indistinguishable by the observed
% corruption. 
This list contains the colors left indistinguishable by the observed
corruption.
Enumerating these lists costs at most $K^{c_{\max}}$
candidate combinations. Let $\mathcal{C}_{\Lambda}(\mathbf{y})$ be the
feasible pinned completions that select each suspect block's color from
its list $\Lambda(v)$. The hard-edge specialization uses the minimum
Hamming subset
$\mathcal{R}_{\Lambda}(\mathbf{y})=
\operatorname*{arg\,min}_{\mathbf{z}\in
\mathcal{C}_{\Lambda}(\mathbf{y})}d_H(\mathbf{y},\mathbf{z})$.
Because hard-edge instances include edge rows in $C$, this selection
still reads only the constraint system.

In soft-edge GCP, adjacency belongs to the objective rather than $C$.
For this encoding, let $\mathrm{conf}_G(\mathbf{z})$ be the number of
monochromatic edges in the coloring represented by $\mathbf{z}$. Define
the minimum-conflict completion set
$\mathcal{R}_G(\mathbf{y})=\operatorname*{arg\,min}_{\mathbf{z}\in
\mathcal{C}_{\Lambda}(\mathbf{y})}\mathrm{conf}_G(\mathbf{z})$, and let
% $\mu_{\mathbf{y}}$ be the implementation's seeded tie-break distribution
% over this set. 
% \begin{equation}
% D_{G}(\mathbf{y})=
% \begin{cases}
% \mathbf{y} & \sigma(\mathbf{y})=\varnothing,\\[2pt]
% \mathbf{z}\sim\mu_{\mathbf{y}},\
% \mathrm{supp}(\mu_{\mathbf{y}})\subseteq\mathcal{R}_G(\mathbf{y})
%   & c(\mathbf{y})\le c_{\max},\\[2pt]
% \bot & \text{otherwise}
% \end{cases}
% \label{eq:decoderG}
% \end{equation}
$\mu_{\mathbf{y}}$ be the implementation's seeded tie-break distribution
with $\operatorname{supp}(\mu_{\mathbf{y}})\subseteq
\mathcal{R}_G(\mathbf{y})$, defined whenever
$\mathcal{C}_{\Lambda}(\mathbf{y})\neq\varnothing$.
% \begin{equation}
% D_G(\mathbf{y})=
% \begin{cases}
% \mathbf{y}, & \sigma(\mathbf{y})=\varnothing,\\[2pt]
% \mathbf{z}\sim\mu_{\mathbf{y}}, &
% \begin{gathered}
% \sigma(\mathbf{y})\neq\varnothing,\quad c(\mathbf{y})\le c_{\max},\\
% \mathcal{C}_{\Lambda}(\mathbf{y})\neq\varnothing,
% \end{gathered}\\[2pt]
% \bot, & \text{otherwise}.
% \end{cases}
% \label{eq:decoderG}
% \end{equation}
\begin{equation}
D_G(\mathbf{y})=
\begin{cases}
\mathbf{y},
  & \text{if }\sigma(\mathbf{y})=\varnothing,\\[3pt]
\mathbf{z}\sim\mu_{\mathbf{y}},
  & \begin{aligned}[t]
    \text{if }&\sigma(\mathbf{y})\neq\varnothing,\\
              &c(\mathbf{y})\le c_{\max},\quad
               \mathcal{C}_{\Lambda}(\mathbf{y})\neq\varnothing,
    \end{aligned}\\[3pt]
\bot,
  & \text{otherwise}.
\end{cases}
\label{eq:decoderG}
\end{equation}
The graph-aware decoder is reported as a separate arm and lies outside
Proposition~\ref{prop:accuracy}: it ranks candidates by graph conflict
and therefore does not sample them uniformly.

% The three implementations can now be compared without forward
% references:

% \begin{center}\small
% \begin{tabular}{lccc}
% \toprule
% decoder & information & budget & selection\\
% \midrule
% generic \syn{}$_{C}$ & $\mathbf{y},C,\mathbf{b}$ & $w_{\max}{=}6$ &
%   uniform on $\mathcal{R}$\\
% hard-edge \syn{}$_{C}$ & $\mathbf{y},C,\mathbf{b}$ & $c_{\max}{=}6$ &
%   uniform on $\mathcal{R}_{\Lambda}$\\
% \syn{}$_{G}$ & $\mathbf{y},C,\mathbf{b},G$ & $c_{\max}$ ($6$/$8$) &
%   minimum $\mathrm{conf}_G$\\
% \bottomrule
% \end{tabular}
% \end{center}

The three implementations can now be compared without forward
references:

\begin{center}
\small
\setlength{\tabcolsep}{3pt}
\renewcommand{\arraystretch}{1.08}

\begin{tabularx}{\columnwidth}{
|>{\raggedright\arraybackslash}X
|>{\centering\arraybackslash}p{0.21\columnwidth}
|>{\centering\arraybackslash}p{0.18\columnwidth}
|>{\centering\arraybackslash}X|}
\hline
\textbf{Decoder} &
\textbf{Information} &
\textbf{Budget} &
\textbf{Selection} \\
\hline

generic \syn{}$_C$ &
$\mathbf{y},C,\mathbf{b}$ &
$w_{\max}{=}6$ &
uniform on $\mathcal{R}$ \\
\hline

hard-edge \syn{}$_C$ &
$\mathbf{y},C,\mathbf{b}$ &
$c_{\max}{=}6$ &
uniform on $\mathcal{R}_{\Lambda}$ \\
\hline

\syn{}$_G$ &
$\mathbf{y},C,\mathbf{b},G$ &
$c_{\max}$ ($6/8$) &
minimum $\mathrm{conf}_G$ \\
\hline
\end{tabularx}
\end{center}

\noindent
In the evaluation, \syn{}$_{C}$ supports the theory, cross-family tests,
the published $K{=}V$ (KV) benchmark suite, and the objective-blind
hard-edge ladder (Section~\ref{sec:eval-blind}). \syn{}$_{G}$ supports the
soft-edge scale ladder and hardware campaign
(Section~\ref{sec:eval-device}).

These definitions establish four properties. First, feasible
measurements pass through unchanged, so before resampling \syn{}'s
candidate pool contains every measured bitstring that purification would
retain. This per-boundary property does not imply end-to-end quality
dominance, which we measure separately. Second, pinning preserves every
variable outside the suspect set, making repair a local edit rather than
a re-solve. Third, \syn{}$_{C}$ is objective-blind. Finally, $\bot$
provides a bounded give-up rule and therefore defines the decoder's own
liveness limit.

For fixed budgets, generic \syn{}$_{C}$ costs $O(n^{w_{\max}}m)$ and the
block-structured variants cost
$O(\mathrm{poly}(n)\cdot K^{c_{\max}})$. Thus the implementation is
bounded by design rather than acting as an unrestricted nearest-feasible
solver. These budgets are fixed for each campaign and never tuned per
cell. They are nonbinding over the deployable range: in the
device-calibrated block profiles, 0.07\% of infeasible measurements
exceed $c_{\max}{=}6$ and none exceed 8; in the generic-decoder profiles,
0.07\% require more than $w_{\max}{=}6$ flips. On 16,000 planted
hard-edge single-flip
trials, the block specialization returns exactly the same repair set as
the ideal decoder; divergence begins only under multi-flip corruption,
so Proposition~\ref{prop:accuracy} remains explicitly a single-flip
result. Campaign claims use the archived stochastic trajectories; the
released artifact pins the Python hash seed for deterministic reruns.

\emph{What the decoder never sees.} The decoder never sees the
objective (for \syn{}$_{C}$; \syn{}$_{G}$'s adjacency use is reported
explicitly), never sees the optimum, and no categorical success is
ever scored on the register: a repair decoder applied to pure noise can
fabricate feasible registers (the solver-contamination trap,
Section~\ref{sec:eval-setup}), so categorical claims count only
measured bitstrings. Stage~3 formalizes this distinction between
system-level quality and measurement-level attribution.

\subsection{Stage 3: deployment regime analysis}
\label{sec:method-death}

\begin{table}[!b]
  \caption{Deployment regimes when the quality crossover
  \(\lambda_q\) is defined and precedes the purification cliff
  \(\lambda_s\). The outer boundary \(\lambda_r\) is \syn{}'s own
  liveness limit.}
  \label{tab:regimes}
  \footnotesize
  \renewcommand{\arraystretch}{1.15}
  \setlength{\tabcolsep}{3pt}
  \centering
  \begin{tabular}{|c|l|c|l|}
    \hline
    regime & condition & prefer & reason \\
    \hline
    I   & $\lambda<\lambda_q$ & \purify{} & no demonstrated repair
      advantage \\
    \hline
    II  & $\lambda_q\le\lambda<\lambda_s$ & \syn{} & both alive; repair
      wins quality \\
    \hline
    III & $\lambda_s\le\lambda<\lambda_r$ & \syn{} &
      \begin{tabular}{@{}l@{}}\purify{} survival ${<}1/2$;\\\syn{} for liveness\end{tabular} \\
    \hline
    IV  & $\lambda\ge\lambda_r$ & --- & repair leaves its own
      envelope \\
    \hline
  \end{tabular}
\end{table}

\begin{figure*}[t]
  \centering
  \includegraphics[width=0.75\textwidth]{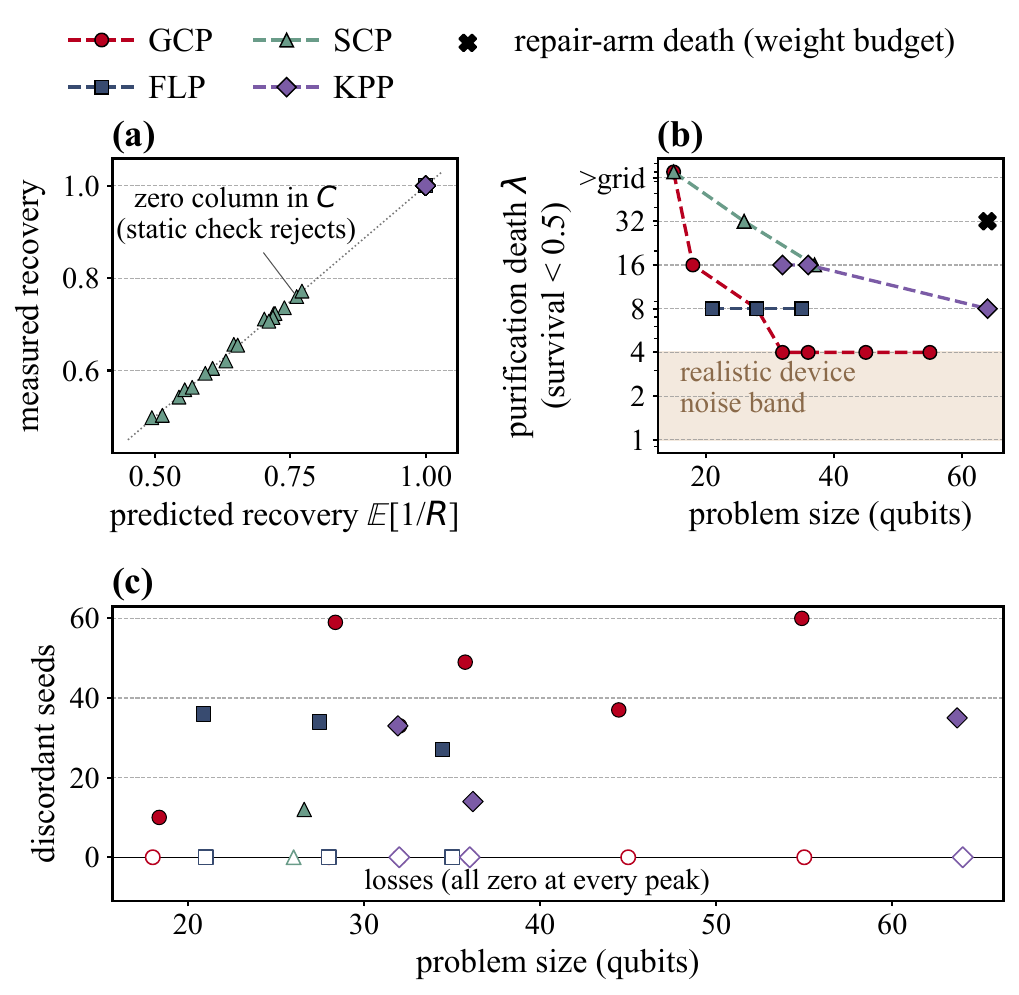}
  \caption{Generalization across four constraint families. (a) Predicted versus measured single-flip repairability. (b) Purification survival cliffs versus problem size. (c) Peak paired optimum-sampling discordance in the post-cliff region. For each scale, the selected noise cell maximizes the number of repair-only successes minus purification-only successes among cells at or above the purification survival cliff. Positive counts are instance--seed pairs in which only \syn{} measures an optimum; negative counts indicate the reverse. Counts pool 12 paired seeds per instance across the instances at each scale. Thirteen scales have a post-cliff cell with nonzero discordance, and no selected peak cell contains a repair loss.}
  \label{fig:e6}
\end{figure*}

For each fixed instance, noise scale \lam{}, and shot budget $S$, the
deployment decision uses three endpoints: full-chain survival,
final-register quality, and optimum sampling from measured bitstrings. We
define them in that order.

\emph{Trajectory-preserving survival} requires the chain to reach
segment $T$ without replacing any boundary output with an external
register. Let $s_t(\Pi)=1$ when policy $\Pi$ uses such a substitution at
boundary $t$, and zero otherwise. Restart, revert, and oracle
reinjection therefore count as substitutions even when they allow
execution to continue:
\begin{equation}
\begin{gathered}
s_t(\Pi)=\mathds{1}\!\left[
  \substack{X_{t+1}\text{ is externally supplied}\\
  \text{rather than derived from }Y_t}\right],\\
\mathrm{surv}(\Pi)=\Pr\Bigl[\text{reach }T \wedge
  \textstyle\sum_{t<T} s_t=0\Bigr].
\end{gathered}
\label{eq:subst}
\end{equation}
Let $M(\Pi)=\biguplus_{t=0}^{T-1}Y_t$ be the multiset union of every
bitstrings measured during a run. Also let
$f^{*}=\min_{\mathbf{x}\in F}f(\mathbf{x})$ and
$\overline{f}(X_T)=|X_T|^{-1}\sum_{\mathbf{x}\in X_T}f(\mathbf{x})$.

\begin{definition}[Regime quantities]
\label{def:regime}
The \emph{survival cliff} of policy $\Pi$ is
$\lambda^{*}(\Pi)=\min\{\lambda : \mathrm{surv}(\Pi)<\tfrac12\}$ on the
tested grid. If survival never falls below $1/2$, the boundary is
right-censored as $\lambda^{*}>\lambda_{\max}$, where $\lambda_{\max}$
is the largest tested scale. For a completed chain with
$f^{*}\neq0$, the \emph{quality endpoint} is the incumbent's
approximation-ratio gap
$\mathrm{ARG}(\Pi)=|\overline{f}(X_T)-f^{*}|/|f^{*}|$. Because $X_T$
may contain repaired bitstrings, ARG measures the complete hybrid runtime
and is not evidence of quantum origin. The \emph{optimum-sampling
endpoint} is
$\Pr[\exists\,\mathbf{y}\in M(\Pi)\cap F: f(\mathbf{y})=f^{*}]$,
which counts only raw measured bitstrings. Restricting this categorical
endpoint to $M(\Pi)$ prevents a decoder-constructed optimum from being
credited to the quantum computation.
\end{definition}

Section~\ref{sec:eval-setup} specifies the statistical tests and multiple-comparison corrections for these endpoints. 
Three noise scales then organize the deployment map:
\begin{equation}
\begin{gathered}
\lambda_q=\text{quality crossover while both policies are live},\\
\lambda_s=\lambda^{*}(\Pi_{\mathrm{P}})\ \text{(purification cliff)},\\
\lambda_r=\lambda^{*}(\Pi_{\syn{}})\ \text{(repair liveness)},
\end{gathered}
\label{eq:lamq}
\end{equation}
Here $\lambda_q$ is the first tested noise scale at which
$\Delta_{\mathrm{ARG}}=\mathrm{ARG}(\Pi_{\mathrm{P}})-
\mathrm{ARG}(\Pi_{\syn{}})>0$ with grid-wide significance and both
policies have survival at least $1/2$. If no tested scale satisfies
these conditions, $\lambda_q$ is undefined. The first two boundaries
select between purification and repair; $\lambda_r$ is the outer
liveness limit of repair itself. All three are empirical onsets on a
pre-specified grid, not closed-form predictions.

On the four KV scales with a significant quality crossover, the measured
ordering is $\lambda_q<\lambda_s<\lambda_r$, producing the four regimes
in Table~\ref{tab:regimes}. When $\lambda_q$ is undefined, regime~II is
absent: purification remains the default below $\lambda_s$, and liveness
becomes the selection criterion at $\lambda_s$. If
$\lambda_q\ge\lambda_s$, regime~II is likewise empty. Regime~IV marks
repair's own liveness boundary, not an optimization-reachability limit;
for some workloads, failure to reach the optimum can bind earlier.

Stage~3 estimates these boundaries by sweeping the noise grid; it is not
a pre-execution predictor. Estimating $\lambda_q$ or $\lambda_s$ for an
unseen instance remains open because constraint structure affects the
cliff at least as much as problem size (Figure~\ref{fig:e6}b). We report
undefined crossovers explicitly and apply the grid-wide correction in
every $\lambda_q$ comparison.

%% =====================================================================

\section{Evaluation}
\label{sec:eval}

We evaluate \syn{} through four questions. \textbf{Q1:} Does the static
profiler predict single-flip repairability across constraint families
(\S\ref{sec:eval-gen})? \textbf{Q2:} Do the end-to-end gains come from repair
rather than mere continuation or classical solving (\S\S\ref{sec:eval-core}--\ref{sec:eval-blind})? 
%first at small
% scale and then with an objective-blind decoder at 78--90 qubits
% (\S\S\ref{sec:eval-core}--\ref{sec:eval-blind})? 
\textbf{Q3:} Where should a
system select repair over purification, and how does that choice scale
(\S\ref{sec:eval-regime})? \textbf{Q4:} Do the locality premise, quality
gain, and liveness gain survive on real hardware, including when the decoder cannot access the objective (\S\ref{sec:eval-device})? 
% Section~\ref{sec:eval-ablate} tests the encoding and runtime choices that support these results.
% update with new experiment
% We also evaluate observable locality, end-to-end behavior, and decoding cost
% on the IBM Heron backend \texttt{ibm\_marrakesh}. Four campaigns consume
% 1{,}136\,s of QPU time: 544\,s and 282\,s for the two graph-aware
% campaigns, 144\,s for the three-arm objective-blind quality experiment,
% and 166\,s for the hard-edge tests and exploratory attempts.
% Figure~\ref{fig:e7} reports the graph-aware \syn{}$_G$ results;
% Figure~\ref{fig:e9} tests
% objective-blind \syn{}$_C$ on the same device.

\subsection{Experimental Setup}
\label{sec:eval-setup}

\textbf{Workloads.} We evaluate 109 instance packs, including small
cross-family instances down to 8 qubits and scale ladders extending to
120 qubits: 31 hard-edge $K{=}V$ GCP packs from Rasengan's published suite~\cite{rasengan},
16 soft-edge $K{=}3$ packs, four fixed-$K$ hard-edge packs, four
go/no-go packs, and 18 packs each for FLP~\cite{facility_fcp}, KPP~\cite{k_partition_kpp}, and SCP~\cite{set_cover_scp}. The static screen rejects one generated SCP instance
with a zero column in $C$; every retained instance passes. SCP has
median repair degeneracy $R{=}2$ and predicted recovery
$\ER{}\approx0.70$, whereas the other families have median $R{=}1$.
All arms use Rasengan's original chains, encodings, training procedure,
circuits, and angles; only the boundary runtime changes.

% \textbf{Execution model and budgets.} Simulation uses the exact
% factorized channels enabled by measurement between segments and the
% small active set of each transition. Under the independent local
% depolarizing-and-readout model, the active subcircuit is simulated as a
% density matrix and spectator qubits undergo independent readout flips.
% The noise multiplier \lam{} scales the base rates as
% \begin{align}
% p_{1q}(\lambda) &= \min\bigl(3.5{\times}10^{-4}\lambda,\; 0.5\bigr),
% \nonumber\\
% p_{2q}(\lambda) &= \min\bigl(8.75{\times}10^{-3}\lambda,\; 0.75\bigr),
% \nonumber\\
% p_{\mathrm{ro}}(\lambda) &= \min\bigl(10^{-2}\lambda,\; 0.5\bigr),
% \label{eq:noisecaps}
% \end{align}
% where $p_{1q}$ and $p_{2q}$ are the one- and two-qubit depolarizing
% parameters and $p_{\mathrm{ro}}$ is the per-qubit readout-flip
% probability; $\lambda=1$ approximates current-device rates. Every
% simulated comparison uses $S\in\{64,256\}$ shots per segment and 12
% paired seeds. All 85 analytically tractable packs pass noise-free
% propagation checks, and the other 24 pass channel-law tests. Up to 24
% qubits, the maximum deviation from whole-register noisy simulation is
% 0.05. ~\ref{app:sim} gives the factorization identity, modeled
% omissions, and validation details.

\textbf{Execution model and budgets.}
Simulation uses factorized local depolarizing and readout channels
enabled by the small transition active sets. A noise multiplier
$\lambda$ scales the calibrated one-qubit, two-qubit, and readout
error rates; $\lambda=1$ approximates current-device rates. We use
$S\in\{64,256\}$ shots per segment and 12 paired seeds. ~\ref{app:sim}
gives the complete channel definition and validation against
whole-register simulation.

% \textbf{Baselines and controls.} The deployable incumbent is
% purification, Rasengan's published boundary runtime~\cite{rasengan}.
% We also construct two stronger continuation baselines that activate
% only when purification returns no survivor. \emph{purify+restart}
% repopulates the register from the feasible seed, whereas
% \emph{purify+last} reuses the previous surviving register; both resume
% at the failed segment. They always complete the chain but introduce a
% substitution event ($s_t=1$, Eq.~\eqref{eq:subst}) whenever activated,
% on average 6.2 times in 26 segments at 24 qubits and 33.4 times in 95
% segments at scale.

\textbf{Baselines and controls.} Purification is the boundary policy in
Rasengan's original implementation~\cite{rasengan}. We introduce two
continuation baselines for this study to test whether replacing an empty
register is sufficient to recover performance. Both first apply purification.
If no feasible shot survives, \emph{purify+restart} initializes the next
segment from the original feasible seed, whereas \emph{purify+last} reuses
the previous feasible register. These fallbacks allow the chain to continue
but introduce an external substitution ($s_t=1$, Eq.~\eqref{eq:subst}).
\syn{} instead derives the next register from the current measurements
through bounded repair.

Two diagnostic controls separate repair from classical solving.
\emph{Coset-uniform} uses the same syndrome and lists $\Lambda(v)$ as
the block decoder but selects uniformly. On soft-edge GCP this removes
adjacency-aware selection; on hard-edge GCP the common feasibility
check still rejects edge violations, so the control isolates whether
edge rows guide candidate selection. The A0 control feeds pure noise to
the same decoder and measures how often the decoder can solve an
instance without circuit-generated structure. Uniform reinjection over
$F$ is an oracle control available only when $F$ is enumerable; it
upper-bounds uninformed continuation but is not deployable.

\textbf{Endpoints and statistical tests.} A0 demonstrates why
categorical endpoints must be measurement-level: when applied to random
bits, the decoder alone produces median register proper-coloring rates
of 0.49 at 24 qubits and 0.79 at 15 qubits. We therefore count a proper
coloring or sampled optimum only when it was measured, never when it
was created by repair. Approximation-ratio gap (ARG) remains the
register-level system metric of Definition~\ref{def:regime}. We report
per-instance medians over 12 paired seeds. Quality comparisons use the
paired Wilcoxon signed-rank test; optimum-sampling discordance uses a
paired binomial test. When chain death leaves ARG undefined, a test
includes only seeds for which both arms are defined, and we report the
pair count when it falls below 12. Attribution tests compare \syn{}
against the five controls within each cell and apply Holm correction
within that planned family. Purification-versus-\syn{} tests answer the
separate deployment question and apply Holm correction across the
tested \lam{} grid.

\begin{figure}[t]
  \centering
  \includegraphics[width=0.75\linewidth]{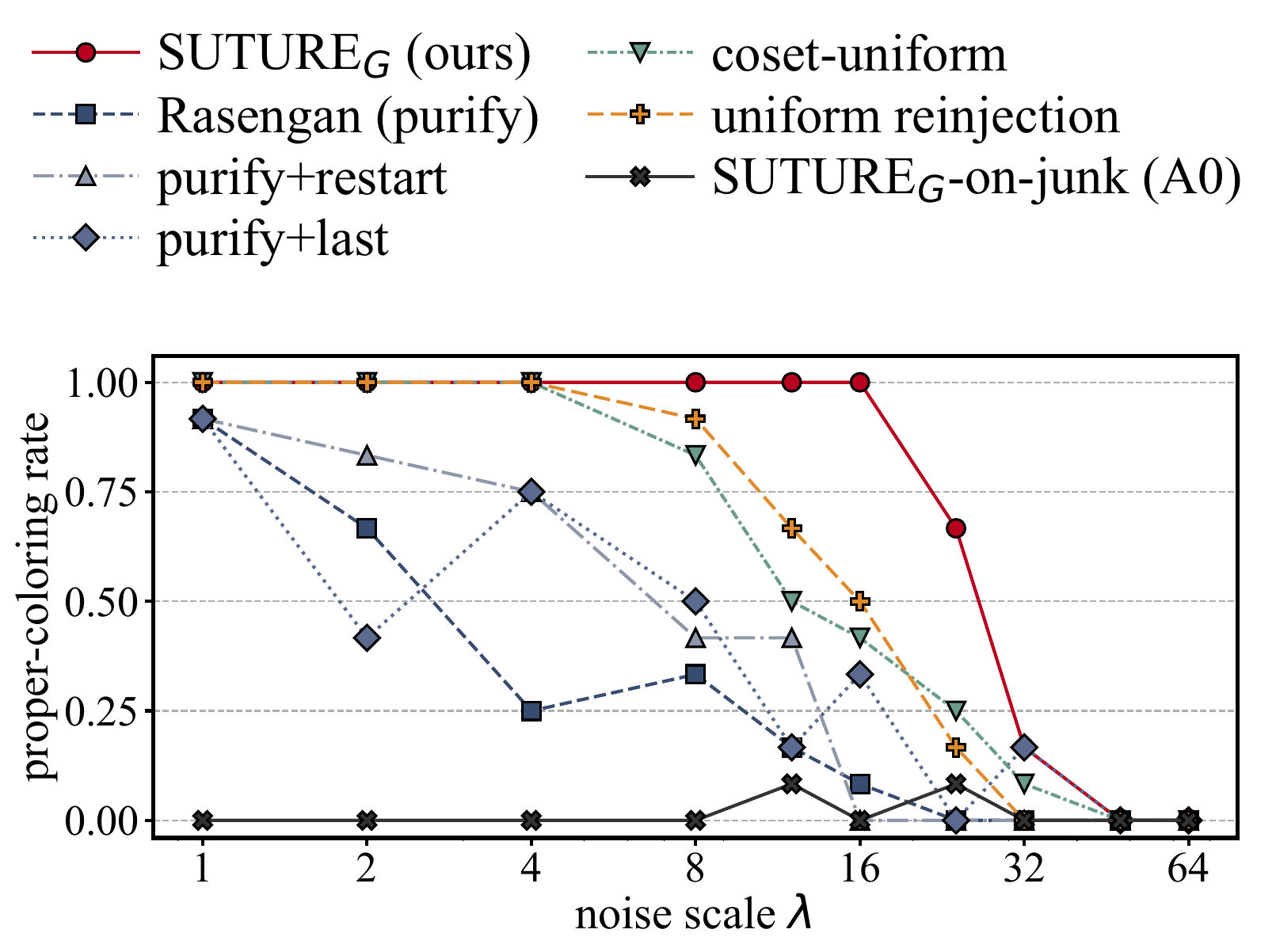}
  \caption{Go/no-go attribution test on 24-qubit GCP (\(S=64\)). Proper-coloring measurement rate across \syn{}\(_G\), purification, fallback, ablation, oracle, and A0 control arms.}
  \label{fig:e1}
\end{figure}

\subsection{Q1: Does the static profiler predict repairability?}
\label{sec:eval-gen}

Figure~\ref{fig:e6}a tests the profiler's single-flip accuracy law
across 70 instances from four constraint families. Measured recovery
differs from the predicted \ER{} by at most 0.011. The profiler also
distinguishes the harder SCP family, where slack-column collisions
reduce \ER{}. For an 8-qubit SCP diagnostic instance, one zero column
predicts a silent-flip probability of 12.5\%; the planted-flip experiment
observes 11.9\%. These results verify both the theorem's implementation and the profiler's ability to identify an unsuitable constraint encoding before execution.

This experiment deliberately tests planted single-bit corruptions; we
do not claim that device errors are predominantly single-bit. The
separate locality premise A2 states only that measured violations
usually involve few syndrome-active blocks, making the bounded block
search computationally useful. A multi-bit corruption may still
activate only one or two blocks. We test A2 under calibrated noise and
on hardware in Section~\ref{sec:eval-device}; the chain experiments
below determine whether the single-flip profile remains informative
under the full multi-flip distribution. Thus Figure~\ref{fig:e6}a
validates D1--D3 as a static repairability characterization, while the
remaining questions test its end-to-end value.

\begin{figure}[t]
  \centering
  \includegraphics[width=0.75\linewidth]{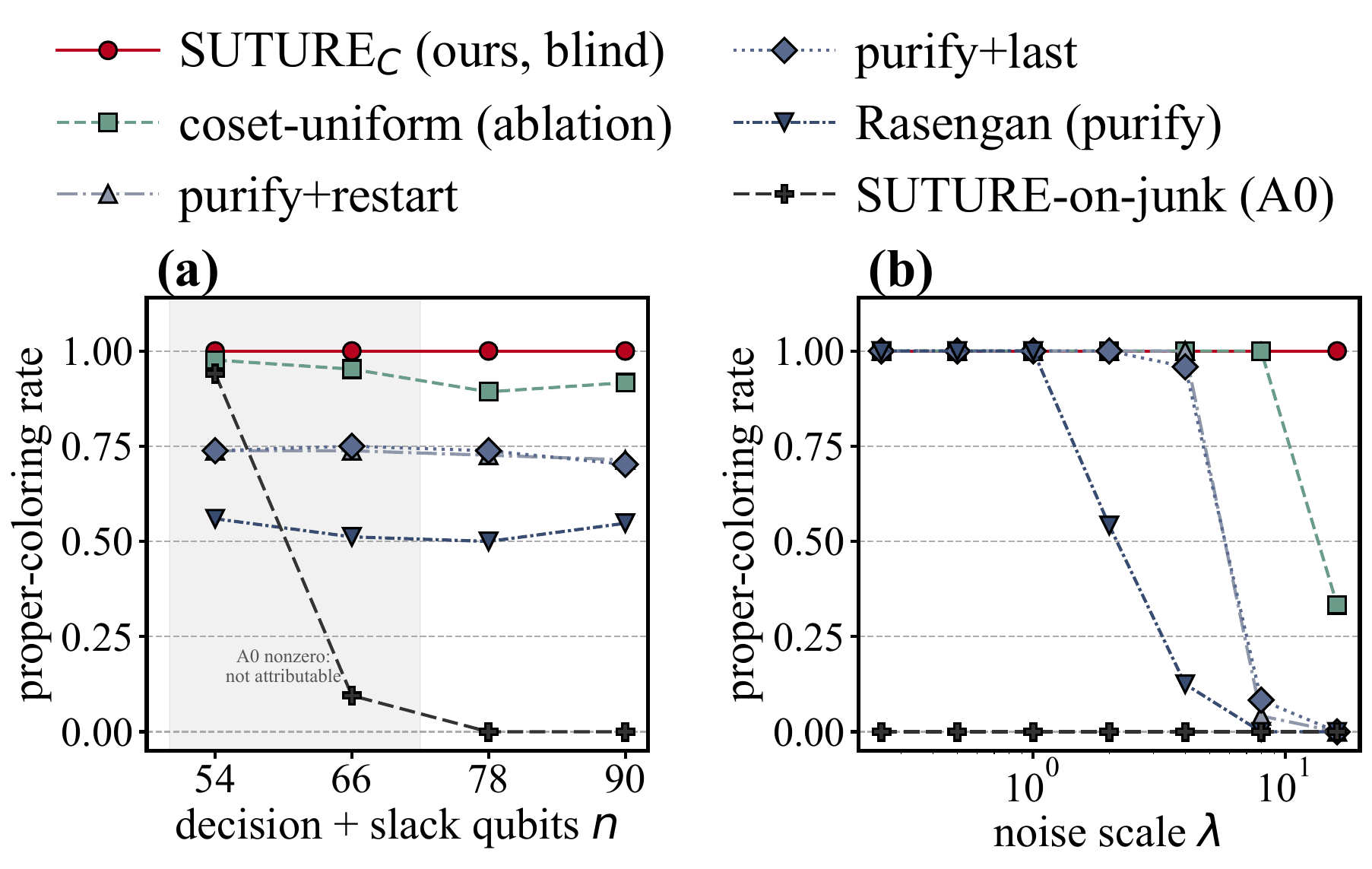}
  \caption{Objective-blind hard-edge ladder (54–90q). (a) Proper-coloring performance; shaded sizes fail the A0 contamination screen. (b) At contamination-screened 78–90q, \syn{}\(_C\) separates from purification and its fallbacks in Holm-corrected post-cliff cells; the coset-uniform ablation separates only at the highest noise.}
  \label{fig:e8}
\end{figure}

\subsection{Q2: Is repair responsible for the gain?}
\label{sec:eval-core}

We begin with a 24-qubit soft-edge GCP instance at $S=64$ shots per
segment (Figure~\ref{fig:e1}). Here \syn{}$_G$ uses adjacency-aware selection, whereas the
coset-uniform control uses the same syndrome and candidate lists
$\Lambda(v)$ (Eq.~\eqref{eq:cosetlist}) but samples uniformly. Their
difference therefore isolates the selection rule rather than syndrome
localization.

% \syn{}$_G$ maintains a proper-coloring measurement rate of 1.0 through
% $\lambda=16$, where purification has fallen to 0.08; purification dies
% by $\lambda=24$.
\syn{}$_G$ maintains a proper-coloring measurement rate of 1.0 through
$\lambda=16$. At that scale, purification measures a proper coloring in
1 of 12 runs, although all purification chains terminate before completing
the full sequence. Purification's survival cliff is $\lambda_s=12$
(5 of 12 chains complete), and its measured proper-coloring rate reaches
zero at $\lambda=24$.
At $\lambda=16$, \syn{}$_G$
outperforms each of the five attribution controls---purify+restart,
purify+last, uniform reinjection, coset-uniform repair, and A0---with
paired $p<0.05$ after Holm correction. Because $\lambda=16$ was the
first death-window cell identified by the same sweep rather than a
pre-specified cell, this result is descriptive. 
% The planned,
% grid-corrected attribution test is the scale ladder in
% Section~\ref{sec:eval-blind}.
Section~\ref{sec:eval-blind} separately evaluates objective-blind repair
using the hard-edge encoding and grid-corrected comparisons.

The controls explain the separation. Both continuation baselines
perform worst despite completing every chain, showing that completion
alone is insufficient. A0 produces a proper coloring in only 1 of 12
seeds at $\lambda=12$ and $\lambda=24$, and in none of the other eight
noise cells. This rate is an order of magnitude below \syn{}$_G$ on
circuit output, so adjacency-aware search cannot reproduce the result
without circuit-generated structure to repair.

% The same conclusion holds when edge constraints are included in $C$.
% A separate 54-qubit hard-edge experiment evaluates best measured objective cost, objective-blind \syn{}$_C$ retains a median best measured cost of 29 after purification collapses at
% $\lambda=2$; every realizable baseline has median cost 39
% ($p=0.0005$). Uniform reinjection is omitted because enumerating $F$
% would make the control unrealizable at this scale.

\begin{figure*}[t]
  \centering
  \includegraphics[width=0.75\textwidth]{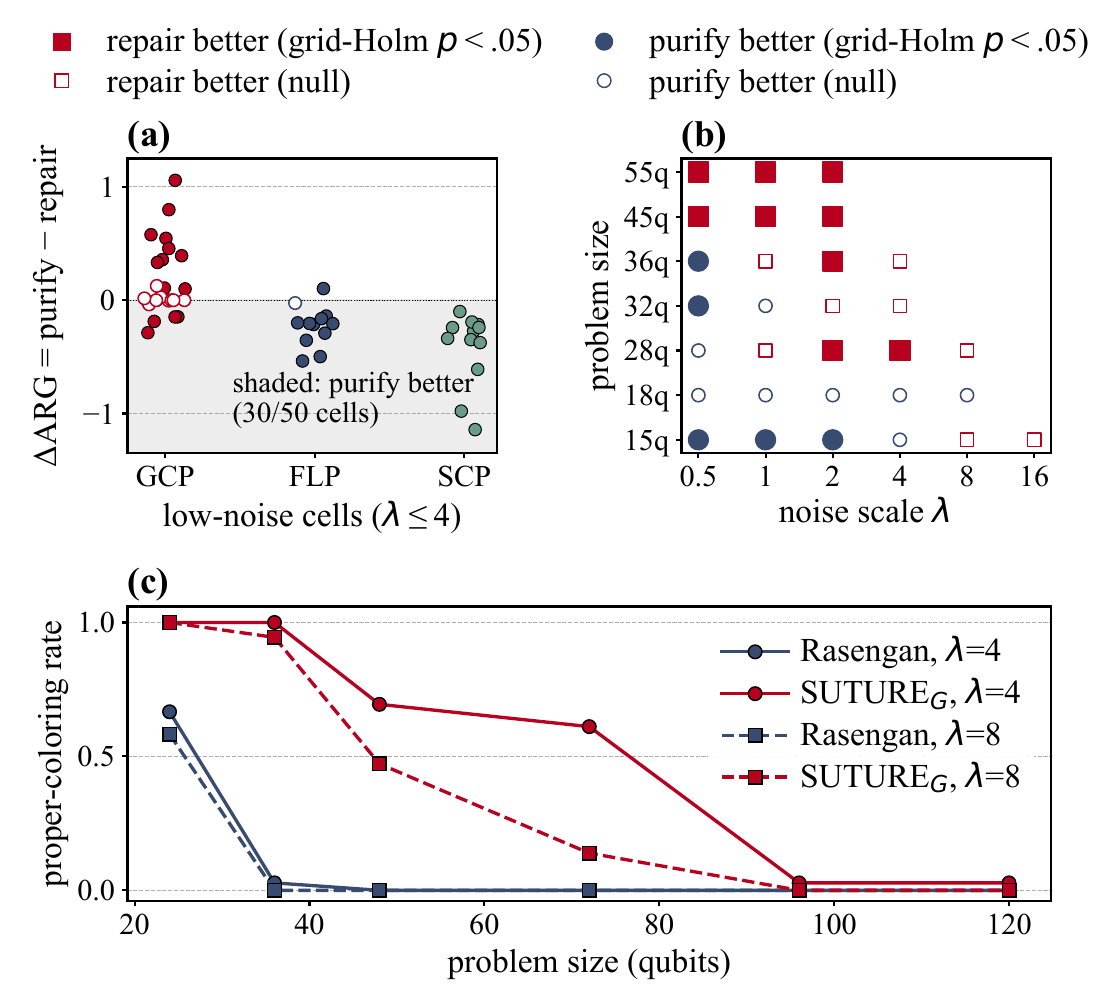}
  \caption{Deployment and scale. (a) Low-noise ARG margins. (b) Grid-corrected quality crossover \(\lambda_q\) on the KV suite. (c) Soft-edge proper-coloring performance of \syn{}\(_G\) versus purification through 120 qubits.}
  \label{fig:q3}
\end{figure*}

\subsection{Q2 continued: Objective-blind attribution}
\label{sec:eval-blind}

The preceding soft-edge experiment uses graph adjacency. To remove that
source of objective information, we construct a hard-edge ladder with
$V\in\{8,10,12,14\}$, $K=3$, and 54--90 qubits including slack.
Edge rows now belong to $C$, while the random vertex--color objective
remains hidden from every boundary runtime. All arms share the same
pre-trained chain and COBYLA training protocol. Transition active sets
remain at 6--8 qubits as $V$ grows, so channel-extraction cost scales
with chain length rather than $|F|$. Because $F$ is exactly the set of
proper colorings, any repair that leaves an edge conflict returns
$\bot$.

A0 determines which scales support attribution. When fed pure noise,
the decoder produces a proper coloring in 94\% of runs at 54 qubits and
10\% at 66 qubits; these scales are contaminated and excluded. At 78
and 90 qubits, A0 produces no proper coloring in 168 runs across two
sizes and seven noise levels, corresponding to a pooled one-sided 95\%
upper bound of 1.8\%. Only these two scales carry attribution evidence
(Figure~\ref{fig:e8}a). This screen is post hoc: A0 was included from
the start, but the zero-observed-success rule was adopted after
examining it. We therefore describe the result as
contamination-screened rather than confirmatory.

At 78 and 90 qubits, \syn{}$_C$ measures a proper coloring in every
cell and has no paired-seed loss to any arm. The 190 raw discordances
are descriptive because seeds recur across the \lam{} grid. The
pre-specified per-cell tests instead treat each size and comparison arm
as one Holm family over seven noise levels. At both sizes, \syn{}$_C$
separates from purification at $\lambda\in\{4,8,16\}$ and from both
continuation baselines at $\lambda\in\{8,16\}$. Aggregating each seed
once over the post-cliff interval $\lambda\in\{4,8,16\}$ gives 24 wins
in 24 seed--instance pairs against purification and each continuation
baseline ($p=1.2{\times}10^{-7}$); this aggregation is secondary to the
per-cell tests.

The coset-uniform comparison is narrower. Across the same 24
seed--instance pairs, \syn{}$_C$ records 16 wins, no losses, and eight
ties ($p=3.1{\times}10^{-5}$). Per-cell separation appears only at
$\lambda=16$: it is significant at 78 qubits
($p_{\mathrm{Holm}}=0.027$) but not at 90 qubits
($p_{\mathrm{Holm}}=0.11$). We therefore conclude that objective-blind
repair outperforms purification and its continuation baselines at both
uncontaminated scales, but do not claim that constraint-guided
candidate selection improves monotonically with size. Consistent with
the regime map, all circuit-fed arms are indistinguishable at
$\lambda\le1$, purification begins losing at $\lambda=2$, and
grid-corrected separation begins when the post-cliff region opens at
$\lambda=4$.

\begin{figure*}[t]
  \centering
  \includegraphics[width=0.75\textwidth]{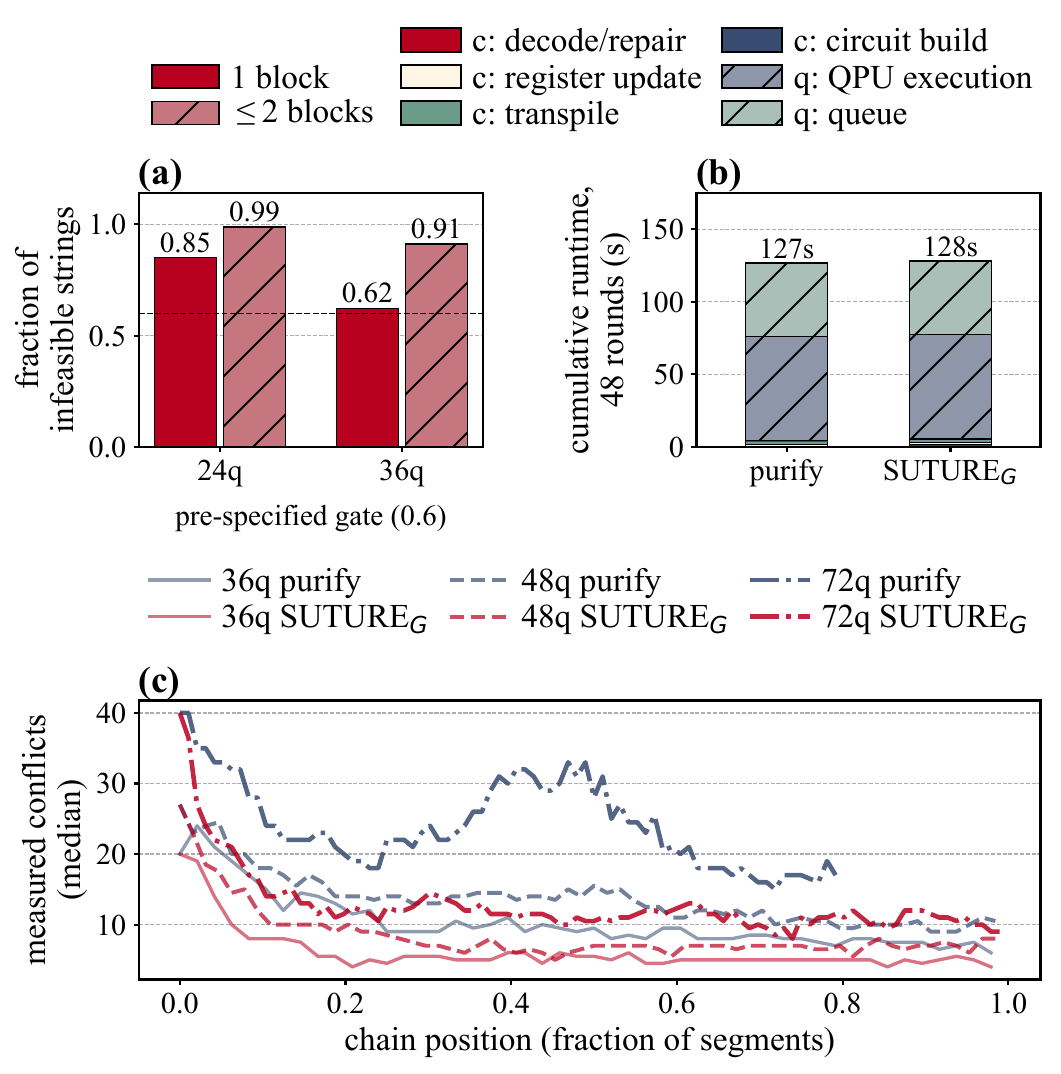}
  \caption{IBM Heron hardware results. (a) Observable block locality at the 24q/36q probe sizes. (b) 36q timing: repair adds 1.25 s, \(<1\%\) of pipeline latency. (c) End-to-end quality and liveness at 36–72q; at 72q and \(S=64\), 11/12 purification chains die while \syn{}\(_G\) completes all runs.}
  \label{fig:e7}
\end{figure*}

\subsection{Q3: When should repair be deployed, and does it scale?}
\label{sec:eval-regime}

Repair should not replace purification uniformly. The relevant decision
is where the quality crossover $\lambda_q$ and liveness boundary
$\lambda_s$ place a workload on the regime map of
Section~\ref{sec:method-death}.

\textbf{Below the crossover, purification is the safer default.}
When both arms survive at $\lambda\le4$, purification has lower median
ARG in 30 of 50 non-degenerate cells across GCP, FLP, and SCP; 26 are
statistically significant, with margins as large as $-0.71$
(Figure~\ref{fig:q3}a). Repair leads in 15 cells, including 12
significant cells at larger scales whose crossover already lies in this
band, and five cells tie. Repair removes purification's discard filter,
so a low-noise quality penalty is expected and must remain part of the
deployment rule.

\textbf{The crossover moves earlier with scale.}
On the paper-fidelity KV suite, repair becomes preferable at 28 qubits
from $\lambda=2$ (ARG margin $+0.11$ to $+0.58$, grid-Holm
$p\le10^{-4}$), at 36 qubits from $\lambda=2$ ($+0.10$,
$p=4{\times}10^{-4}$), and at 45--55 qubits from $\lambda=0.5$
($+0.33$ to $+1.06$, $p<10^{-3}$; Figure~\ref{fig:q3}b). At 32
qubits, the largest observed margin is $+0.55$ at $\lambda=4$, but it
does not survive grid correction (within-cell $p=0.03$,
$p_{\mathrm{Holm}}=0.09$); no $\lambda_q$ is therefore defined on that
grid. At 55 qubits purification never samples an optimum, even at
$\lambda=0$, whereas repair reaches optimum-sampling rates of
0.83--1.0 for $\lambda=0.5$--$4$.

% \textbf{The liveness gain extends across families and to 120 qubits.}
% On the fixed-$K$ soft-edge ladder, neither purification nor a
% reinjection baseline measures a proper coloring from 48 qubits onward,
% whereas \syn{}$_G$ continues the chain of liveness up to 120 qubits, measured proper-coloring success becomes rare at the largest sizes
% (Figure~\ref{fig:q3}c). A0 succeeds in 0.019 of runs at 24 qubits and
% never at 36 qubits or above; coset-uniform repair also reaches zero from
% 48 qubits. Thus adjacency-aware selection carries this particular
% soft-edge result beyond 48 qubits. In contrast, the KV and cross-family
% ARG and survival results use objective-blind \syn{}$_C$. Across the 13
% scales from four families with a post-cliff window, the peak-discordance
% cell contains no repair loss (Figure~\ref{fig:e6}c). This is one
% predefined summary cell per window, not a claim about every cell.
\textbf{Liveness extends to 120 qubits; solution success remains size-dependent.}
On the fixed-$K$ soft-edge ladder, neither purification nor a
reinjection baseline measures a proper coloring from 48 qubits onward
at the noise scales shown in Figure~\ref{fig:q3}c. At $\lambda=4$ and $8$,
\syn{}$_G$ completes every chain through 120 qubits, but at both 96 and
120 qubits its proper-coloring measurement rate is only $1/36$ for
$\lambda=4$ and zero for $\lambda=8$. Chain completion therefore extends
beyond the range of reliable proper-coloring success. A0's run-level
probability of at least one proper-coloring measurement is 0.019 at
24 qubits and zero at 36 qubits or above; coset-uniform repair also
reaches zero from 48 qubits. Thus adjacency-aware selection improves
measured proper-coloring success at intermediate sizes and sustains
this soft-edge execution through 120 qubits. In contrast, the KV and
cross-family ARG and survival results use objective-blind \syn{}$_C$.
Across the 13 scales from four families with nonzero optimum-sampling
discordance in a post-cliff window, the selected peak-discordance cell
contains no repair loss (Figure~\ref{fig:e6}c). This is one selected
summary cell per window, not a claim about every cell.

\textbf{More shots postpone, but do not remove, purification death.}
% On the 24-qubit go/no-go instance, increasing the segment budget from
% $S=64$ to $S=256$ moves purification's death from $\lambda=12$ to
% $\lambda=24$. Over the same budgets, \syn{}$_G$ measures proper
% colorings through $\lambda=16$ and $\lambda=24$, respectively, and
% retains the larger operating envelope. The shift follows
% Eq.~\eqref{eq:cliffprob}: additional shots buy liveness linearly rather
% than changing the underlying survival mechanism.
On the 24-qubit go/no-go instance, increasing the segment budget from
$S=64$ to $S=256$ moves purification's survival cliff from
$\lambda_s=12$ to $\lambda_s=24$. All purification chains have terminated
by $\lambda=16$ and $\lambda=24$, respectively. Over the same budgets,
\syn{}$_G$ maintains a proper-coloring measurement rate of 1.0 through
$\lambda=16$ and $\lambda=24$, respectively. These are not its last
successful noise scales: its measured success rate remains nonzero at
$\lambda=32$ for $S=64$ and at $\lambda=48$ for $S=256$.
Additional shots postpone purification's survival cliff, consistent with
Eq.~\eqref{eq:cliffprob}.

% \textbf{Secondary mechanism observation.}
% Optimum sampling under repair is non-monotone in noise. For an
% undertrained chain, moderate noise supplies moves outside the trained
% trajectory and repair restores feasibility; the effect persists for
% training budgets from 0 to 3000 iterations. Because this sweep varies
% only training budget and reuses the go/no-go and ladder controls, we
% treat the result as a mechanism observation rather than a primary
% claim.

\textbf{Repair also has an outer boundary.}
Within the KV and cross-family sweep, repair deaths occur only for 64-qubit
$k$-partition at $\lambda\in\{32,64\}$: survival is 1.0 through
$\lambda=16$ and zero for every seed at the two higher scales. Here the
minimum repair weight exceeds $w_{\max}$, so the generic decoder returns
$\bot$ instead of searching an exponential space. This observed
$\lambda_r$ is a fixed-budget limit, not a tuned failure point; each
campaign uses the same decoder budget in every cell.

% \subsection{Design ablations}
% \label{sec:eval-ablate}

% The encoding ablation tests whether saving qubits is worth weakening
% the syndrome. Figure~\ref{fig:e5}, domain-wall encoding reduces a representative instance
% from 24 to 16 qubits, but four of its six possible single-bit flips
% move between valid codewords and are therefore silent. Under the
% i.i.d.\ flip channel, its measured detectable-flip rate is 0.33,
% compared with 0.94 for one-hot encoding. End to end, domain-wall
% \syn{}$_G$ reaches proper-coloring rates of 0.46--0.54, statistically
% indistinguishable from domain-wall purification at approximately 0.5.
% One-hot \syn{}$_G$, by contrast, reaches 0.92 where one-hot purification
% reaches zero. This result supports D1's design rule: repairability
% depends on syndrome visibility, not qubit count alone.

% % Two runtime knobs produce null results at the tested power of 12 seeds
% % and $v\le24$. Transition reordering yields paired
% % $p=0.75$--$0.90$, while the ensemble-prior tie-break yields $p=0.29$
% % and $0.45$. We therefore retain neither modification in the reported
% % system.

\subsection{Q4: Does it survive real hardware?}
\label{sec:eval-device}

\begin{figure}[t]
  \centering
  \includegraphics[width=0.75\linewidth]{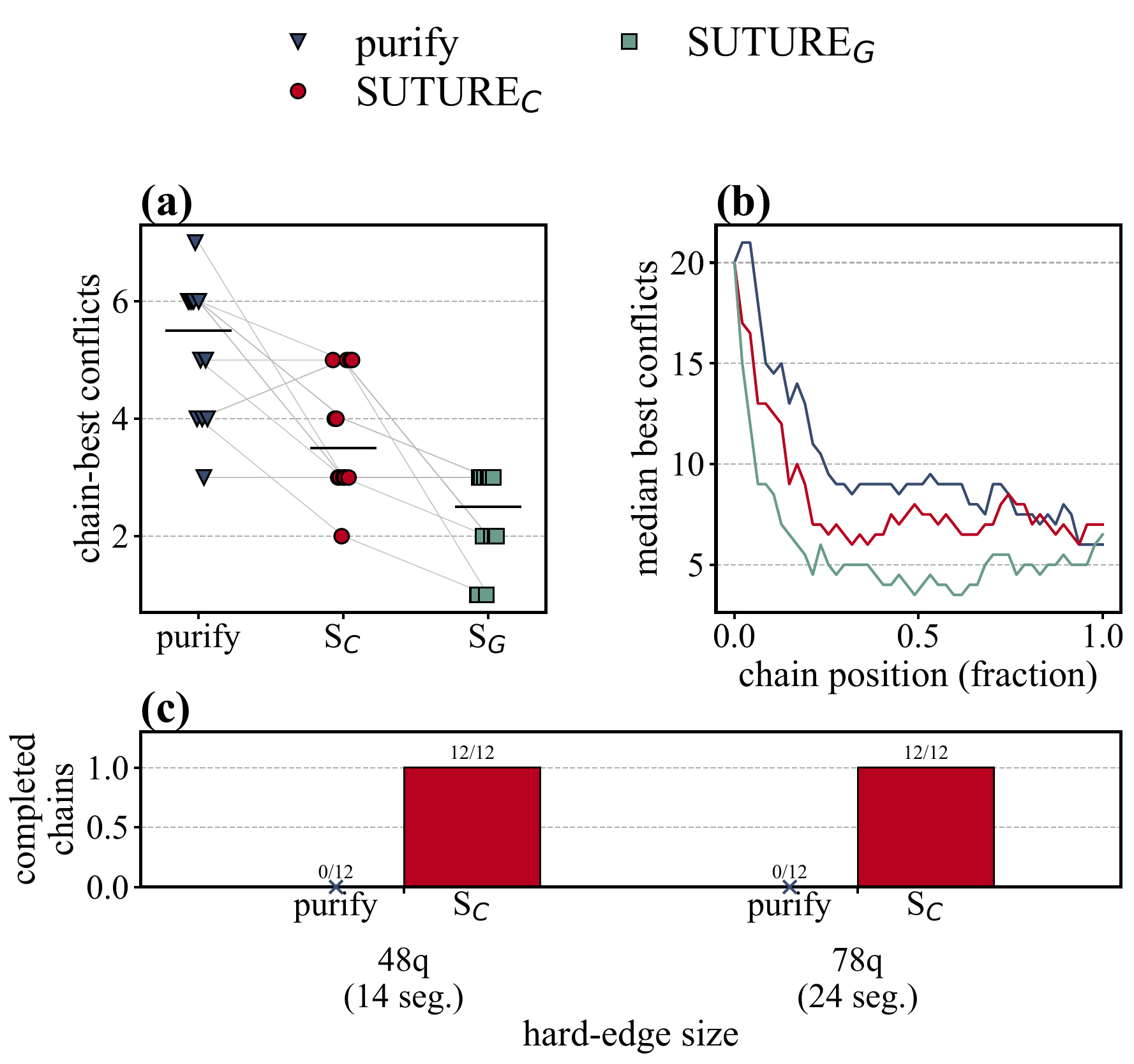}
  \caption{Objective-blind repair on IBM Heron. \textbf{(a,b)} At 36 qubits, all arms complete 48 segments.
  \textbf{(c)} At 48 and 78 qubits,
  purification terminates at the first boundary, whereas \syn{}$_C$
  completes every chain. 
  % Panel (c) reports liveness, not quality.
  }
  % \Description{Three-panel hardware comparison. Panels (a) and (b) show
  % lower measured conflict counts for objective-blind and graph-aware SUTURE
  % than for purification at 36 qubits. Panel (c) shows zero of twelve
  % purification chains and twelve of twelve objective-blind SUTURE chains
  % completing the hard-edge workloads at 48 and 78 qubits.}
  \label{fig:e9}
\end{figure}

We evaluate observable locality, end-to-end behavior, and decoding cost
on the IBM Heron backend \texttt{ibm\_marrakesh}~\cite{IBMQuantum2026}. Four campaigns consume
1{,}136\,s of QPU time: 544\,s and 282\,s for the two graph-aware
campaigns, 144\,s for the three-arm objective-blind quality experiment,
and 166\,s for the hard-edge tests and exploratory attempts.
% We evaluate the locality premise and end-to-end regime map on the IBM
% Heron backend \texttt{ibm\_marrakesh}~\cite{IBMQuantum2026}. 
% The two pre-specified campaign
% allowances consume 826\,s of QPU time (544\,s and 282\,s).
Figure~\ref{fig:e7} reports the graph-aware \syn{}$_G$ results;
Figure~\ref{fig:e9} tests
objective-blind \syn{}$_C$ on the same device.
% These soft-edge experiments
% use the same graph-aware \syn{}$_G$ implementation as simulation,
% including adjacency access; the objective-blind attribution result
% remains the separate hard-edge ladder of
% Section~\ref{sec:eval-blind}.

\textbf{Observable-locality gate.}
The probe samples trained-chain segments from known feasible inputs and
counts syndrome-active blocks directly from each infeasible
measurement. At 24 and 36 qubits, respectively, 85\% and 62\% of
infeasible bitstrings activate exactly one block, while 99\% and 91\%
activate at most two (Figure~\ref{fig:e7}a). Both sizes pass the
pre-specified 0.6 gate on the single-block fraction. At 72 qubits, the
fractions fall to 28\% for exactly one block and 62\% for at most two;
the probe-size gate is not met. Consequently, the 72-qubit experiment
tests the full multi-block decoder rather than the single-block fast
path. This measurement validates only the observable
$c(\mathbf{y})$ used by the search; it does not reconstruct latent
physical-error support.

\textbf{Graph-aware quality and liveness.}
Figure~\ref{fig:e7}c, before the death boundary, \syn{}$_G$ wins the paired register-quality
comparison in 10 of 12 seeds at 36 qubits ($p=0.011$) and all 12 seeds
at 48 qubits ($p=0.0005$). At 72 qubits and $S=64$, 11 of 12
purification chains die between segments 7 and 84 (median 76), whereas
\syn{}$_G$ completes all 96 segments in every seed (discordance 11/0,
$p=0.00098$). At its most difficult rounds, the decoder repairs a
median 59 of 64 shots even when the measured feasible fraction reaches
zero. All seeds remain scoreable for measured conflict count:
\syn{}$_G$ achieves a median chain-best value of 5 (range 3--8), versus
15 for purification (range 12--32), and wins all 12 pairs
($p=0.0005$).

The lower shot budget exposes the same survival mechanism more sharply.
At $S=32$, purification's median death moves to segment 4; one third of
seeds never measure a feasible bitstring, and every purification chain is
dead by segment 24, while \syn{}$_G$ completes the chain (discordance
12/0, $p=0.00049$). Thus the device data exhibit the $1/S$ dependence
predicted by Eq.~\eqref{eq:cliffprob}. Neither runtime measures a proper
coloring at 24 vertices, so optimum-sampling claims remain
simulation-scoped.

\textbf{Runtime cost.}
Over the 48-round, 36-qubit chain, decoding adds 1.25\,s
(Figure~\ref{fig:e7}b). This is less than 1\% of the 128\,s pipeline
latency including queueing. Excluding the variable 50.5\,s queue delay,
decoding is 1.6\% of the 77.6\,s execution path, which is dominated by
72\,s of QPU time. In this campaign, in-loop repair is therefore small
relative to quantum execution while providing the observed liveness
gain.

\textbf{Objective-blind quality on hardware.}
Figures~\ref{fig:e9}a and~\ref{fig:e9}b, we additionally run purification, \syn{}$_C$, and \syn{}$_G$ in the
same 36-qubit hardware jobs, with interleaved seeds to share
calibration drift. All arms complete 48 segments. Median chain-best
measured conflicts are 5.5, 3.5, and 2.5, respectively; both repair
variants significantly improve over purification after Holm
correction. Thus \syn{}$_C$ establishes an objective-blind hardware
gain, while \syn{}$_G$ quantifies the additional benefit of
graph-aware selection.

\textbf{Objective-blind hard-edge liveness.}
The soft-edge experiment establishes objective-blind quality, but its
constraint matrix does not include edge rows. We therefore also test
hard-edge \syn{}$_C$, for which edge feasibility is part of the
syndrome rather than an objective query.
% This encoding makes each transition much deeper: multi-controlled logic
% over an eight-qubit active set transpiles to 842 two-qubit gates at depth
% 2{,}149 on the reported pinned layout, compared with two two-qubit gates
% for a soft-edge transition. 
We run 12 paired seeds with $S=64$ at 48 qubits for 14 segments and at 78 qubits for 24 segments.
Neither circuit-fed arm measures a feasible bitstring in either cell:
0 of 11{,}520 device-measured shots at 48 qubits and 0 of 19{,}200 at
78 qubits. Purification consequently terminates at the first boundary
for all 12 seeds, and no measurement-level quality endpoint is defined.
\syn{}$_C$ nevertheless completes every segment for all 12 seeds at
both sizes (survival discordance 12/0; exact paired sign test
$p=0.00049$ at each size). At 78 qubits, it repairs a median 25 of 64
shots per round. Thus, bounded objective-blind repair sustains execution
when discard-based enforcement has no feasible measurement to retain.
Figure~\ref{fig:e9}c summarizes the two cells. This is strictly a liveness
result: we draw no quality or optimum-sampling claim from either hard-edge
hardware cell.

%% =====================================================================
\section{Discussion}
\label{sec:discussion}

The applicability of \syn{} is structural rather than universal.
The interface is general: any segmented optimizer with linear
constraints can compute a syndrome from measured bitstrings. The benefit,
however, depends on the constraint matrix. D1--D3 separate
detectability, localization, and repair ambiguity before execution,
while the regime map determines whether repair is preferable at a
given noise level. This distinction explains the cross-family results.
FLP and KPP open clear high-noise windows: the selected post-cliff
optimum-sampling cells reach 36/0 and 35/0 paired wins/losses,
respectively (Figure~\ref{fig:e6}c). SCP opens a narrower window: its
slack-dense rows both preserve enough feasible mass to make purification
more resilient and increase repair ambiguity. The same structure can therefore make
one runtime stronger without making the other easier to replace.
 
Measurement is a systems resource,
\syn{} depends on segmentation, not merely on a constraint-preserving
ansatz. Measurement between segments turns the register into a
collection of per-shot classical states that the decoder can inspect,
repair, and reinject. In a monolithic commuting-Hamiltonian QAOA
pipeline~\cite{chocoq}, the same
operation becomes terminal post-processing because no intermediate
boundary exists at which repaired state can affect subsequent
execution. Segmentation is therefore part of the recovery interface:
it trades uninterrupted quantum evolution for observable checkpoints
at which the system can act.
 
Deployment requires two different predictions,
the static profiler estimates whether local corruption is detectable
and recoverable; it does not predict the quality crossover
$\lambda_q$ or the purification cliff $\lambda_s$ for an unseen
instance. Those boundaries remain empirical in this work. 
Predicting these boundaries from constraint structure,
circuit depth, and device calibration is the main remaining step toward
an automatic runtime selector.
The simulation uses a factorized-channel model that omits crosstalk, spectator preparation and idle errors, and correlated readout, we validate this model using whole-register checks where simulation remains tractable (~\ref{app:sim}). For larger 24--72-qubit instances, real-hardware experiments capture the aggregate system-level effects of these omitted error sources under realistic device conditions. 
%Hardware results are collected on a single backend across multiple days for each problem size.

\section{Related Work}
\label{sec:related}
% Classical constraint repair~\cite{minton1992minconflicts, Hoos2000} is the classical repair heuristic for constraint-satisfaction problems that repeatedly selects a conflicted variable and assigns a value minimizing the remaining constraint violations. SUTURE shares the general intuition that violated constraints contain useful repair information, but serves a different systems role. Rather than iteratively searching for a CSP solution, SUTURE uses the violated rows of the compiled constraint matrix as a syndrome to localize a bounded per-shot repair at an intermediate VQA boundary, pins variables outside the suspect set, and reinjects the repaired state into the continuing quantum execution. SUTURE\(_C\) further selects feasible repairs without accessing the optimization objective, while SUTURE\(_G\) separately uses graph-conflict information for candidate ranking.
Min-conflicts~\cite{minton1992minconflicts} is a classical repair heuristic for constraint-satisfaction problems that repeatedly selects a conflicted variable and assigns a value minimizing the remaining constraint violations. More broadly, stochastic local-search algorithms have been studied extensively for SAT~\cite{Hoos2000}. SUTURE shares the general intuition that violated constraints contain useful repair information, but serves a different systems role. Rather than iteratively searching for a CSP solution, SUTURE uses the violated rows of the compiled constraint matrix as a syndrome to localize a bounded per-shot repair at an intermediate VQA boundary, pins variables outside the suspect set, and reinjects the repaired state into the continuing quantum execution. SUTURE\(_C\) further selects feasible repairs without accessing the optimization objective, while SUTURE\(_G\) separately uses graph-conflict information for candidate ranking.
 
Rejection and subspace recovery,
symmetry verification discards measurements outside a known symmetry
sector~\cite{symverify,mcardle-symmetry}. Purification applies the same
policy to feasibility, so both lose liveness when the expected survivor
count per boundary approaches one. Sample-based quantum diagonalization
(SQD) instead projects symmetry-violating configurations toward a target
subspace~\cite{sqd, sqd-codespace}. \syn{} shares the goal of recovering
useful samples, but uses the violated problem constraints as a syndrome
and reinjects each repair before the next segment. The distinction is
therefore both informational, problem constraints localize the
corruption, and temporal, recovery changes the remaining execution.

Terminal repair and classical improvement,
post-processing methods repair samples after quantum execution.
REGRID-QAOA performs objective-aware greedy descent~\cite{regrid-qaoa};
Ising-machine workflows similarly combine repair with
improvement~\cite{ising-bench}; annealers resolve broken
chains~\cite{chainbreak-qst,qac-me}; and hybrid solvers add classical
objective heuristics~\cite{dupont-greedy}. These methods can improve
terminal solutions but do not sustain a segmented chain, and their
objective access complicates attribution. The core \syn{}$_C$ path is
instead objective-blind and acts at every boundary. DQI also uses a
decoder~\cite{dqi}, but there decoding implements the algorithmic
transform itself rather than recovering noisy outputs. To our
knowledge, none of these approaches performs constraint-aware,
per-shot reconstruction inside a continuing feasibility-preserving
execution.

Readout mitigation and noise-aware execution, 
readout-error mitigation estimates and inverts calibrated confusion
models over output distributions~\cite{nation-readout,bravyi-readout,
tannu-bias}. Noise-adaptive compilation reduces error through mapping
and scheduling~\cite{murali-noiseadaptive,tannu-variability,ravi-cloud},
while zero-noise extrapolation targets expectation
values~\cite{zne}. These approaches improve the channel, distribution,
or estimated observable. \syn{} addresses a different systems failure:
an individual infeasible sample would otherwise be discarded at a
segment boundary, so it reconstructs a feasible state for immediate
reuse. The techniques are complementary because a better physical
channel can reduce how often repair is invoked.
% 
% Circuit-level constraint handling,
% penalty formulations~\cite{lucas-ising}, constraint-respecting
% mixers~\cite{xy-mixers,wang-xy,grover-mixers,fuchs-constraint-mixers},
% commute-Hamiltonian designs~\cite{chocoq}, and Rasengan's transition
% Hamiltonians~\cite{rasengan} progressively move feasibility into the
% circuit construction. \syn{} operates at the next systems boundary:
% even when ideal dynamics remain in $F$, noise can produce infeasible
% measurements that a segmented runtime must handle. We therefore leave
% Rasengan's circuits, encodings, and training unchanged and replace only
% its discard policy.

Relation to quantum error correction,
quantum error correction (QEC) introduces engineered redundancy and
decodes physical syndromes to protect logical
state~\cite{fowler2012surface,google2024decoder}. \syn{} does not
provide that protection and is not a substitute for QEC. It instead
uses redundancy already present in the optimization constraints, without
adding code qubits, to recover useful classical states at measurement boundaries.
% boundaries. QEC reduces the production of corrupted samples; \syn{}
% defines how a constrained runtime can use the residual samples that
% still violate problem structure.
\syn{} turns constraints that were previously used only to construct feasible quantum dynamics into runtime metadata for recovering and continuing noisy segmented execution.

\section{Conclusion}
\label{sec:conclusion}

% Segmented feasibility-preserving VQAs currently handle noisy,
% infeasible measurements by discarding them. We show that the same
% constraint matrix used to preserve ideal dynamics can also serve as a
% runtime recovery interface. \syn{} combines a static repairability
% profile, an in-loop syndrome decoder, and a regime map that separates
% quality and liveness decisions. The profiler's single-flip prediction
% holds across four constraint families; objective-blind \syn{}$_C$
% remains effective through 90 qubits in contamination-screened
% post-cliff tests; graph-aware \syn{}$_G$ extends the soft-edge study to
% 120 qubits; and the quality and liveness regimes persist on IBM Heron
% hardware through 72 qubits. In the 36-qubit timing experiment, decoding
% adds less than 1\% to total pipeline latency.
% The result is not a blanket replacement for purification. When a
% quality crossover is measured, discard remains the conservative choice
% below it and repair improves quality above it. Once purification reaches
% its survival cliff, repair preserves execution until it reaches its own
% fixed budget boundary. The systems implication is that problem
% constraints need not remain only compiler input: at measurement
% boundaries, they can become runtime metadata for recovering and
% continuing noisy constrained computations.

\syn{} turns problem constraints into runtime metadata for bounded
repair and reinjection at segmented VQA boundaries. Across simulation
and IBM hardware, it sustains execution beyond purification's
survival limit while improving measured quality in the demonstrated
regimes. Its benefit remains encoding-, decoder-, and noise-dependent,
providing a practical but explicitly bounded recovery interface for
noisy segmented quantum execution.

\ack{This work was partly supported by Institute for Information \& communications Technology Planning \& Evaluation (IITP) grant funded by the Korea government (MSIT) (No. 2020-0-00014, A Technology Development of Quantum OS for Fault-tolerant Logical Qubit Computing Environment) and Creation of the quantum information science R\&D ecosystem(based on human resources) through the National Research Foundation of Korea(NRF) funded by the Korean government (Ministry of Science and ICT(MSIT)) (No. RS-2023-00256050).}

\funding{Institute for Information \& communications Technology Planning \& Evaluation, Grant ID: 2020-0-00014 and National Research Foundation of Korea, Grant ID: RS-2023-00256050}
% This section is a list of funder names and grant numbers

\roles{S.S. conceived the original idea, designed the methodology, implemented the algorithms, conducted the experiments, analyzed the results, and wrote the manuscript. L.H., S.L. provided guidance on data visualization, contributed feedback on the experimental results, and assisted during manuscript preparation. Y.H. supervised the project, provided scientific guidance, suggested improvements to the methodology and manuscript, and contributed to the overall development of the work. All authors reviewed and approved the final manuscript.}
% List author names and the contributions made to the article, using terms from the NISO Contributor Roles Taxonomy (CRediT) https://credit.niso.org

\data{The source code supporting this study is publicly available via GitHub, with version 1.0.0 released at \url{https://github.com/QCL-PKNU/SUTURE.git}.}

\appendix

\section{Proofs of Propositions~\ref{prop:detect}--\ref{prop:accuracy}}
\label{app:proofs}

All three are direct computations on the syndrome map
$\sigma(\mathbf{y})=\{r:(C\mathbf{y})_r\neq b_r\}$ of
Eq.~\eqref{eq:syndrome}.

\begin{proof}[Proof of Proposition~\ref{prop:detect} (Detection)]
Let $\mathbf{x}\in F$ and $\mathbf{y}=\mathbf{x}\oplus\mathbf{e}_j$.
For any row $r$,
$(C\mathbf{y})_r-(C\mathbf{x})_r = C_{rj}(1-2x_j)$, which is nonzero
iff $C_{rj}\neq 0$. Since $C\mathbf{x}=\mathbf{b}$, row $r$ is violated
by $\mathbf{y}$ iff $r\in\mathrm{supp}(C_j)$; hence
$\sigma(\mathbf{y})=\mathrm{supp}(C_j)$ exactly. If no column of $C$ is
zero this is nonempty for every $j$, so no single flip is silent and
$\delta_{\mathrm{silent}}=0$. (Domain-wall GCP violates the premise ---
adjacent codewords sit at Hamming distance 1 --- which is why its
single flips are silent, ~\ref{app:eval-ablate}.)
\end{proof}

\begin{proof}[Proof of Proposition~\ref{prop:localize} (Localization)]
By the previous computation $\sigma(\mathbf{y})=\mathrm{supp}(C_j)$,
so the suspect set is
$L(\sigma(\mathbf{y}))=\bigcup_{r\in\mathrm{supp}(C_j)}\mathrm{supp}(C_r)$, which
contains $j$ (every row in $\mathrm{supp}(C_j)$ has $j$ in its
support). Its size is at most
$|\mathrm{supp}(C_j)|\cdot\max_r|\mathrm{supp}(C_r)|$ --- column degree
times maximum row width, a static property of $C$. Localization
degrades exactly when a variable sits in wide rows: a
cardinality-style row of width $w$ contributes $w$ suspects.
\end{proof}

\begin{proof}[Proof of Proposition~\ref{prop:accuracy} (Accuracy law)]
The true flip is always a restoring candidate: flipping $j$ again
returns $\mathbf{x}\in F$, and $j\in L(\sigma(\mathbf{y}))$ by
Proposition~\ref{prop:localize}. The decoder draws uniformly from the
$R(\mathbf{y})$ restoring candidates
(Eq.~\eqref{eq:d1}, D3), so conditional on the corrupted bitstring it
recovers the planted flip with probability exactly $1/R(\mathbf{y})$;
taking the expectation over the planting distribution gives
$\Pr[\text{exact recovery}]=\ER{}$. A competing candidate $j'\neq j$
restores feasibility iff its residual cancels the corruption's:
$C\mathbf{y}-\mathbf{b}=s_jC_j$ with $s_j=1-2x_j$, and flipping $j'$
adds $s'_{j'}C_{j'}$ with $s'_{j'}=1-2y_{j'}$, so $j'$ restores iff
$s'_{j'}C_{j'}=-s_jC_j$ --- a \emph{signed column collision}. For the nonnegative
one-hot assignment rows this reduces to $C_{j'}=C_j$ with the opposite
flip direction (the weight-0/weight-2 coset cases of one-hot GCP are the
canonical instance); the signed form is the one that holds in general,
and our facility-location and set-cover systems do carry $-1$
coefficients. Thus $R-1$
counts signed column collisions within the suspect set, and the law is
exact --- validated empirically to max deviation 0.011 across 70
instances (Figure~\ref{fig:e6}a).
\end{proof}

\section{Factorized Simulation: Validity and Checks}
\label{app:sim}
Simulation uses the exact
factorized channels enabled by measurement between segments and the
small active set of each transition. Under the independent local
depolarizing-and-readout model, the active subcircuit is simulated as a
density matrix and spectator qubits undergo independent readout flips.
The noise multiplier \lam{} scales the base rates as
\begin{align}
p_{1q}(\lambda) &= \min\bigl(3.5{\times}10^{-4}\lambda,\; 0.5\bigr),
\nonumber\\
p_{2q}(\lambda) &= \min\bigl(8.75{\times}10^{-3}\lambda,\; 0.75\bigr),
\nonumber\\
p_{\mathrm{ro}}(\lambda) &= \min\bigl(10^{-2}\lambda,\; 0.5\bigr),
\label{eq:noisecaps}
\end{align}
where $p_{1q}$ and $p_{2q}$ are the one- and two-qubit depolarizing
parameters and $p_{\mathrm{ro}}$ is the per-qubit readout-flip
probability; $\lambda=1$ approximates current-device rates.

\begin{figure}[t]
  \centering
  \includegraphics[width=0.75\linewidth]{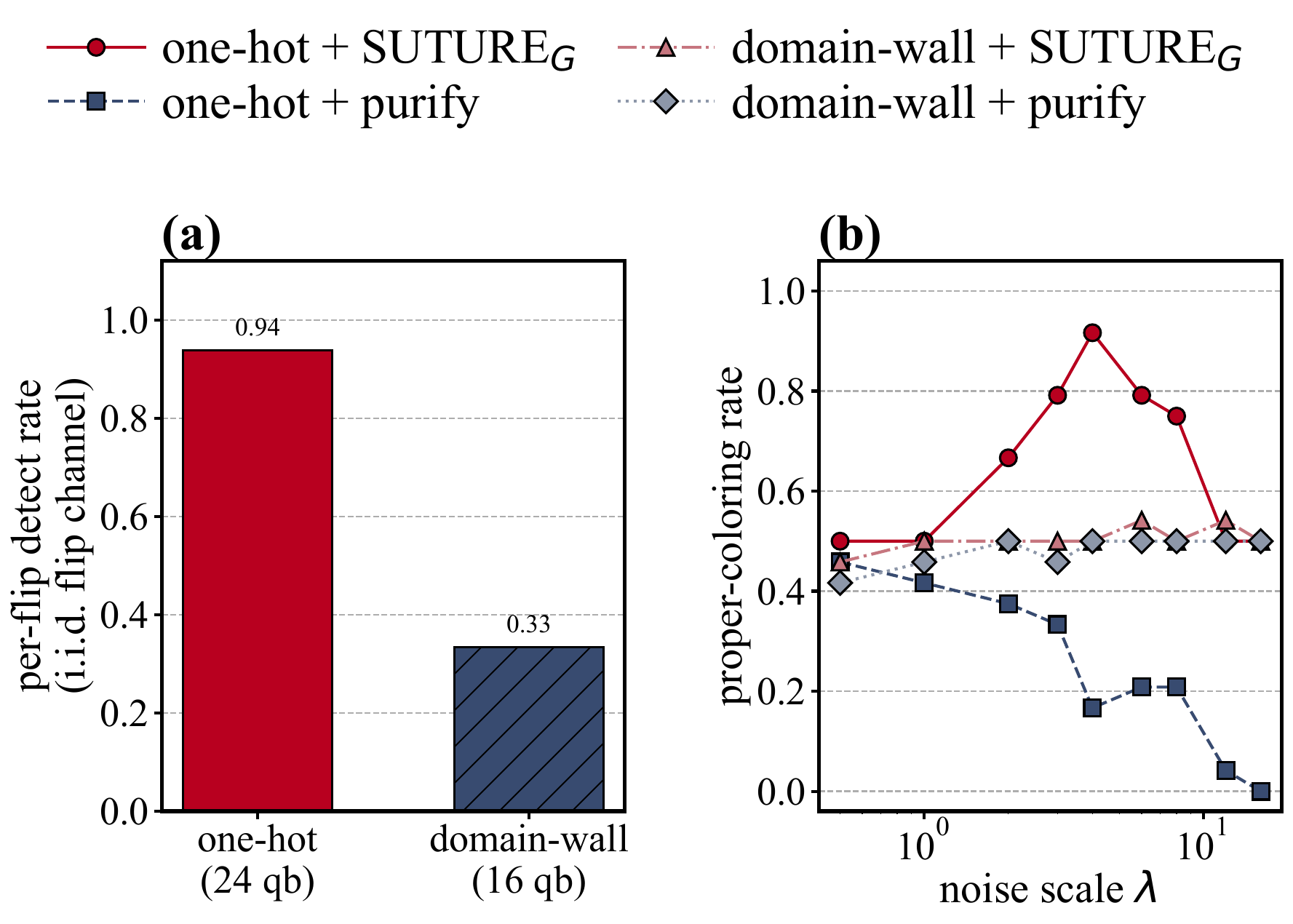}
  \caption{Encoding co-design. (a) Domain-wall reduces qubit count from \(K\) to \(K-1\) per vertex but lowers syndrome detectability. (b) End-to-end, repair provides no gain under domain-wall encoding, while one-hot encoding retains a clear repair advantage.}
  \label{fig:e5}
\end{figure}

Each segment begins from a measured computational-basis bitstring, and its
transition circuit touches only a small active set \(A\). Under the
independent local noise model of Section~\ref{sec:eval-setup}, the
segment output therefore factorizes as
\begin{equation}
P(\mathbf{y}\mid\mathbf{x}) =
T_t(\mathbf{y}_A\mid\mathbf{x}_A)
\prod_{i\notin A}R_i(y_i\mid x_i),
\label{eq:factorize}
\end{equation}
where \(T_t\) is the noisy channel of segment \(t\)'s active subcircuit
and \(R_i\) is the readout-confusion channel for spectator qubit \(i\).
Intermediate measurement is essential: it removes quantum correlations
at every segment boundary, while circuit locality keeps \(T_t\) small.
Equation~\eqref{eq:factorize} is consequently exact under the stated
noise model, rather than an approximation to a full-register state.

The model applies gate depolarization to the active subcircuit and
independent readout flips to every qubit. Explicit \(X\) gates used to
prepare active bits are included in the noisy subcircuit. Reset error,
spectator preparation and idle error, crosstalk, and correlated readout
are not modeled. These are omissions of the noise model, not errors
introduced by the factorization; the hardware campaign tests their
aggregate effect on the observable locality and regime claims.

The implementation first simulates each 2--12-qubit active subcircuit
and stores its channel \(T_t\). It then represents the full register as
classical \(n\)-bit strings, samples active bits from \(T_t\), samples
spectators from the \(R_i\), and applies the selected boundary runtime.
The quantum cost is exponential only in \(|A|\); the full-chain cost is
linear in shots and segments with polynomial per-shot decoder work. The
hard-edge $K{=}V$ suite stops at 55 qubits because its largest active
channels require 12-qubit density matrices and a roughly 0.8\,GB
channel table.
The soft-edge transitions always use two active qubits, allowing the
same procedure to reach 120-qubit problem instances without constructing
a \(2^{120}\)-dimensional state.

We validate both the channel extraction and the sampled chain. At
\(\lambda=0\), all 85 packs for which exact analytic propagation is
tractable retain 100\% feasibility and satisfy the pre-specified ARG
tolerance; sampled-to-analytic ARG deviations are at most 0.009 on the
KV suites and 0.032 on the cross-family packs. The remaining 24 packs
pass tests for row stochasticity, noise-free unitarity, and null-space
validity; the hard-edge ladder also retains feasibility in all 12 seeds
at \(\lambda=0\). Finally, for instances up to 24 qubits, the
factorized sampler and whole-register noisy simulation produce
feasible-fraction medians of 0.68 versus 0.725 at 15 qubits and 0.559
versus 0.525 at 24 qubits. Their maximum observed deviation is 0.05,
traceable to the documented qubit-layout difference; their one- and
two-block corruption profiles differ by at most 0.035 and 0.009,
respectively.

\section{Design ablations}
\label{app:eval-ablate}

The encoding ablation tests whether saving qubits is worth weakening
the syndrome. In Figure~\ref{fig:e5}, domain-wall encoding reduces a representative instance
from 24 to 16 qubits, but four of its six possible single-bit flips
move between valid codewords and are therefore silent. Under the
i.i.d.\ flip channel, its measured detectable-flip rate is 0.33,
compared with 0.94 for one-hot encoding. End to end, domain-wall
\syn{}$_G$ reaches proper-coloring rates of 0.46--0.54, statistically
indistinguishable from domain-wall purification at approximately 0.5.
One-hot \syn{}$_G$, by contrast, reaches a proper-coloring measurement
rate of 0.92 at $\lambda=4$, compared with 0.17 for one-hot purification;
at $\lambda=16$, the corresponding rates are 0.50 and zero.
The initial increase with noise is consistent with more corrupted blocks
becoming eligible for adjacency-aware minimum-conflict repair, whose
outputs seed later quantum segments; it reflects the complete hybrid
runtime and does not establish that stronger noise improves quantum
optimization. This result supports D1's design rule: repairability
depends on syndrome visibility, not qubit count alone.

% Two runtime knobs produce null results at the tested power of 12 seeds
% and $v\le24$. Transition reordering yields paired
% $p=0.75$--$0.90$, while the ensemble-prior tie-break yields $p=0.29$
% and $0.45$. We therefore retain neither modification in the reported
% system.

% \section*{References}
\bibliographystyle{iopart-num-long}
\providecommand{\newblock}{}

\end{document}